%% file: ppkc.tex
\documentclass[review,onefignum,onetabnum]{siamonline190516}

\usepackage{stmaryrd}
\usepackage{dsfont}
\usepackage{enumerate}
\usepackage{graphicx,color}
\usepackage{subfigure} 
\usepackage{amsfonts}
\usepackage{amssymb}
\usepackage{multirow}
\usepackage{bm}

\newsiamthm{claim}{Claim}
\newsiamremark{remark}{Remark}
\newsiamremark{hypothesis}{Hypothesis}
\crefname{hypothesis}{Hypothesis}{Hypotheses}

\usepackage{paralist}
\usepackage{amsopn}

\usepackage{xspace}
\usepackage{bold-extra}
\usepackage[most]{tcolorbox}

\colorlet{texcscolor}{blue!50!black}
\colorlet{texemcolor}{red!70!black}
\colorlet{texpreamble}{red!70!black}
\colorlet{codebackground}{black!25!white!25}

\lstdefinestyle{siamlatex}{%
	style=tcblatex,
	texcsstyle=*\color{texcscolor},
	texcsstyle=[2]\color{texemcolor},
	keywordstyle=[2]\color{texemcolor},
	moretexcs={cref,Cref,maketitle,mathcal,text,headers,email,url},
}

\tcbset{%
	colframe=black!75!white!75,
	coltitle=white,
	colback=codebackground, 
	colbacklower=white, 
	fonttitle=\bfseries,
	arc=0pt,outer arc=0pt,
	top=1pt,bottom=1pt,left=1mm,right=1mm,middle=1mm,boxsep=1mm,
	leftrule=0.3mm,rightrule=0.3mm,toprule=0.3mm,bottomrule=0.3mm,
	listing options={style=siamlatex}
}

\newtcblisting[use counter=example]{example}[2][]{%
	title={Example~\thetcbcounter: #2},#1}

\newtcbinputlisting[use counter=example]{\examplefile}[3][]{%
	title={Example~\thetcbcounter: #2},listing file={#3},#1}

\DeclareTotalTCBox{\code}{ v O{} }
{ 
	fontupper=\ttfamily\color{black},
	nobeforeafter,
	tcbox raise base,
	colback=codebackground,colframe=white,
	top=0pt,bottom=0pt,left=0mm,right=0mm,
	leftrule=0pt,rightrule=0pt,toprule=0mm,bottomrule=0mm,
	boxsep=0.5mm,
	#2}{#1}

\patchcmd\newpage{\vfil}{}{}{}
\begin{tcbverbatimwrite}{tmp_\jobname_header.tex}
	\title{Penalized Projected Kernel Calibration for Computer Models\thanks{Submitted to the editors DATE.
			\funding{Dr. Wang's research was supported by the National Natural Science 
Foundation of China (12101024), the Natural Science Foundation
of Beijing Municipality (1214019).}}}

	\author{Yan Wang \thanks{School of Statistics and Data Science, Faculty of Science, Beijing University of Technology, Beijing 100124, China (\email{yanwang@bjut.edu.cn}).
	}	
	}

	\headers{Penalized Projected Kernel Calibration}{ Yan Wang }
\end{tcbverbatimwrite}
\input{tmp_\jobname_header.tex}

\ifpdf
\hypersetup{ pdftitle={Guide to Using  SIAM'S \LaTeX\ Style} }
\fi

\begin{document}

\maketitle
\begin{abstract}
{ Projected kernel calibration is a newly proposed frequentist calibration method, 
which is asymptotic normal and semi-parametric. Its loss function is usually referred to 
as the {\em PK loss function}. } In this work, we prove the uniform convergence of PK loss function and show that  (1) when the sample size is large, any local minimum point and local maximum point of the $L_2$ loss between the true process and the computer model is a local minimum point of the PK loss function;  (2) all the local minima of the PK loss function converge to the same value. These theoretical results imply that it is extremely hard for the projected kernel calibration to identify the global minimum of the $L_2$ 
loss, i.e. the optimal value of the calibration parameters. To solve this problem, a frequentist method which we term {\em penalized projected kernel calibration method} 
is suggested and analyzed in detail.  We prove that the proposed method is as 
efficient as the projected kernel calibration method.  Through an extensive set of numerical simulations, and a real-world case study, we show that the proposed 
calibration method can accurately estimate the calibration parameters. We also show that its performance compares favorably to other calibration methods regardless of the sample size.
\end{abstract}



\section{Introduction}

Computer models, or simulators, are increasingly used to reproduce the behavior of complex systems in physics, engineering and human processes. Computer models usually involve model parameters that cannot be determined, or observed in the physical processes. The input values of these model parameters may significantly affect the accuracy and  usefulness of the simulations' outputs.  When physical observations are available, one can adjust the computer model parameters so that the computer outputs match the physical data. This stage is called the \emph{calibration} of computer models, and the parameters are usually referred to as \emph{calibration parameters}.
 
The celebrated Bayesian calibration method of Kennedy and O'Hagan \cite{kennedy2001bayesian} is one of the most widely used approaches 
for the calibration of computer models. 
We refer to \cite{kennedy2001bayesian} for more details about computer model calibration. 

{Let us denote the input domain of the physical experiments by $\Omega$, which 
is assumed to be a convex and compact subset of {$\mathbb R^d$.}} Let $\bm X=\{\bm{x}_1,
\ldots,\bm x_{n}\}\subseteq\Omega $ be the set of design points for the physical experiments and $\bm{Y}=(y_1,\ldots,y_n)^{T}$ the corresponding physical responses. 
Throughout the paper,  we denote matrices and vectors
by bold symbols,  with the superscript $T$ indicating transposition. 
Suppose the physical experimental observation is generated by
\begin{eqnarray}
y_i=\zeta(\bm x_i)+\epsilon_i,
\label{eq2.1}
\end{eqnarray}
where $i=1,\ldots,n$; $\zeta(\cdot)$ is called the \emph{true process}, which is an unknown function;  $\epsilon_i$'s  are
independent and identically distributed random variables with mean zero and finite variance  $\sigma^2>0$. { We also assume that the
random error  $\epsilon_i$'s are sub-Gaussian in the sense that there exists universal constants $C_0, \sigma_0>0$
 such that 
 $P\left( |\epsilon_i|>t\right)\leq C_0e^{-t^2/\sigma_0^2}$ holds for all $t > 0$.}

{Let $y^s(\bm x,\bm\theta)$  be the output of the deterministic computer simulator, 
given the control variable $\bm x$ and the calibration parameter $\bm\theta$.  
Suppose $\bm\theta\in\Theta$ and assume that the parameter space $\Theta$ is a compact subset of {$\mathbb R^q$}.} The idea of calibration is to find the values of the calibration parameters such that the computer outputs are as close as possible to the physical experimental observations.

Since computer models are usually built using assumptions and simplifications which are not exactly correct in practice, Kennedy and O'Hagan (KO thereafter) assume that the computer outputs cannot perfectly fit the physical experimental observations. That is, there is an unavoidable \emph{model discrepancy} between the true process and the computer model. KO use the following model (named as KO's model) to take this model discrepancy into account:
\begin{eqnarray}
\zeta(\cdot)= y^s(\cdot,\bm\theta^*)+\delta^*(\cdot),
\label{eq2.2}
\end{eqnarray}
where $\bm\theta^*$ denotes the combination of the optimal calibration parameters and $\delta^*$ is the  \textit{model discrepancy function}. The goal of calibration is to estimate $\bm\theta^*$. However, equation (\ref{eq2.2}) is not enough to fully determine $\bm\theta^*$, because the function $\delta^*$ is also unknown.  

Upon assuming that the model discrepancy function is a realization of a Gaussian Process, KO give a Bayesian estimator of the calibration parameter.  {Tuo and Wu \cite{tuo2015efficient} have shown that the estimator suggested by KO does not converge to  $\bm\theta^*$. The optimal calibration parameters $\bm\theta^*$ is
defined as }
\begin{eqnarray}
\begin{aligned}
\bm\theta^{*}:=\operatorname*{argmin}_{\Theta} L_2(\bm\theta),
\end{aligned}
\label{thetastar}
\end{eqnarray}
where 
\begin{eqnarray}
L_2(\bm\theta)=\int_{\Omega}   ({\zeta(\bm x)-y^s(\bm x,\bm\theta)})^2 d\bm{x}
\label{l2loss}
\end{eqnarray}
is the $L_2$ loss function between the true process and the computer model. 
 
Many efforts have been made to obtain a consistent estimator of the calibration parameter.
Tuo and Wu \cite{tuo2015efficient} have plugged the kernel ridge regression of $\zeta$ 
into (\ref{thetastar}) to get the $L_2$ estimator. Wong et al. \cite{wong2014frequentist} have proposed a least square estimator. Gu and Wang \cite{gu2018scaled} have pointed out that both the approaches have limitations in their calibration and prediction performance, especially when the number of physical observations is small.  To predict the true process model more {accurately}, Gu and Wang \cite{gu2018scaled} have proposed a scaled Gaussian Process model to mimic the model discrepancy function and gave a new Bayesian calibration method based on KO's model, termed scaled Gaussian Process model calibration method. {The convergence rate of  the estimator} given by scaled Gaussian Process model calibration method has been proven to be slower than $n^{-1/2}$  \cite{gu2018theoretical} at variance with the $L_2$  and least square calibrations. This slower convergence rate may be an issue when an efficient estimator of the calibration parameter is needed.
In order to better estimate the calibration parameter, Tuo \cite{tuo2019adjustments} { has} suggested a projected kernel calibration method based on adding an orthogonality constraint (\ref{or cond}) to the KO's model. Tuo has also proven the semi-parametric consistency of this estimator. 

The projected kernel calibration method has been successfully applied to the calibration 
of the composite fuselage simulation  with a small sample  \cite{wang2020effective}.
However, the numerical simulations reported in \cite{gu2018scaled} have shown that it may get stuck in local minima or even local maxima of the $L_2$ loss, especially when the sample size is large. 
 

In this work, we first verify analytically the results obtained numerically  in \cite{gu2018scaled}, and then { propose} a new calibration method which is more robust with respect to the number of physical observations. The main contribution 
of this work may be summarized as follows:

At first, the performance of the projected kernel calibration method is explored and { the results are}:
\begin{itemize}
 \item Any local minimum point and local maximum point of the $L_2$ loss function is a local minimum point of the projected kernel loss function (\ref{pkloss});
 \item All the local minimum values  of the projected kernel loss function converge to the same value as the sample size tends to infinity.
 \item The projected kernel loss function uniformly converges to a projected kernel $L_2$ loss function (\ref{con.delta.2});
\end{itemize}
These results imply that it is extremely hard for the projected kernel calibration to identify the global minima of the $L_2$ loss function.

Second, in order to address the drawbacks of the projected kernel calibration, we put forward a penalized projected kernel calibration method. 
{The main idea of the proposed method is to add a penalty $\|\delta^*\|_{L_2(\Omega)}$ to the projected kernel loss function. By adding this penalty,  the proposed loss function avoids the problems that exist in the projected kernel loss function. That is,  the local minimum points  of the proposed loss function no longer contain the local maximum points  of the $L_2$ loss function; and the local minima of the proposed loss function are different  even for a sufficiently large $n$.} We also prove that the proposed method is  semi-efficient. {A theoretical comparison} between the proposed method and some existing calibration methods is shown in Table \ref{tab:com.calibration}. 

\begin{table}[!htbp]
    \caption{Comparison of different calibration methods}
    \label{tab:com.calibration}
    \centering
    \footnotesize
    \setlength{\tabcolsep}{3pt}
    \renewcommand{\arraystretch}{1.2}
  \tiny{  \begin{tabular}{lccccccc}
        \hline
        Method&Abbreviation &Consistence (Y/N) &Penalty&Consistence rate $(n^{-t})$\\
        \hline
        KO's calibration \cite{kennedy2001bayesian,tuo2020improved} &KO & $N$ &  $\|\delta^*\|_{\mathcal{N}_{K}(\Omega)}$& $m/(4m+d)$  \\
         $L_2$ calibration  \cite{tuo2015efficient}    & $L_2$&$Y$&$\|\delta^*\|_{\mathcal{N}_{K}(\Omega)}$& $1/2$\\
       Least square calibration  \cite{wong2014frequentist} & LS&$Y$&$\|\delta^*\|_{\mathcal{N}_{K}(\Omega)}$& $1/2$\\
      Scaled Gaussian Process model calibration  \cite{gu2018scaled} &SGP & $Y$ &  $\|\delta^*\|_{\mathcal{N}_{K}(\Omega)}$ and $\|\delta^*\|_{L_2(\Omega)}$  & $m/(2m+d)$\\
  Projected kernel calibration  \cite{tuo2019adjustments} & PK & $Y$  & $\|\delta^*\|_{\mathcal{N}_{K_{ \theta^*}}(\Omega)}$ & $1/2$\\
   Penalized projected kernel calibration &    PPK&$Y$&$\|\delta^*\|_{\mathcal{N}_{K_{\theta^*}}(\Omega)}$ and $\|\delta^*\|_{L_2(\Omega)}$  & $1/2$\\
           \hline
    \end{tabular}}
  {Remark 1:  
  $\|\cdot\|_{L_2(\Omega)}$, $\|\cdot\|_{\mathcal{N}_{K}(\Omega)}$ and $\|\cdot\|_{\mathcal{N}_{K_{\theta^*}}(\Omega)}$ denote the corresponding the $L_2$ norm,  the norm of the  reproducing kernel Hilbert space  generated by $K$ and $K_{\theta^*}$ respectively. 
  
   Remark 2: Suppose the  reproducing kernel Hilbert space $\mathcal{N}_{K}(\Omega)$ can be continuously embedded into a (fractional) Sobolev space $H^m(\Omega)$.}
  \end{table}

The  paper is organized as follows. In Section \ref{sec2}, we give a brief review of the projected kernel calibration method.  In Section \ref{comparison},  the performance of the projected kernel calibration method {is examined} theoretically and a numerical 
simulation is  conducted to verify these theoretical assertions.  In Section \ref{sec4}, a penalized projected kernel calibration method is introduced and its convergence properties are analyzed. Computational problems are addressed in Section \ref{comp}. The results of two numerical examples and a real-world spot welding study are presented in Section \ref{sec5}. Concluding remarks and further discussions are given in Section \ref{sec6}.

\section{Review on projected kernel calibration}
\label{sec2}
In this section we review the projected kernels and the projected kernel calibration.

\subsection{Projected kernels}
\label{pkernel}
Let $K(\cdot,\cdot)$ to be a positive definite kernel function over $\Omega\times\Omega$, such as
the Mat\'ern kernel function \cite{santner2013design,stein2012interpolation} with
\begin{eqnarray}\label{matern f}
K(h;\nu,\rho)=\frac{1}{\Gamma(\nu)2^{\nu-1}}\left(\frac{ h}{\rho}\right)^\nu K_\nu\left(\frac{ h}{\rho}\right), \nu>0, \rho>0,
\end{eqnarray}
where $h=\parallel \bm{x}_i-\bm{x}_j\parallel$ and $\|\cdot\|$ denotes the usual Euclidean distance; $K_\nu$ is the modified Bessel function of the second kind; $\nu$ and $\rho$ are \textit{fixed} parameters.
Suppose $\mathcal{G}$ is a finite-dimensional subspace of $ L_2(\Omega)$ with $dim \mathcal{G} = q$ and  $\{e_1,\ldots,e_q\}$ is a set of orthonormal basis of $\mathcal{G}$. For any $f \in L_2(\Omega)$, let $\mathcal{P}_{\mathcal{G}}f:=\sum_{i=1}^q <f,e_i>_{L_2(\Omega)}e_i$ be the projection of $f $ onto $\mathcal{G}$ and 
$\mathcal{P}^{\perp}_{\mathcal{G}}f=f-\mathcal{P}_{\mathcal{G}}f $ be the perpendicular component. Then the projected kernel of $K$ is defined as
\begin{eqnarray}
 K_{\mathcal{G}}=K-\mathcal{P}_{\mathcal{G}}^{(1)}K-\mathcal{P}_{\mathcal{G}}^{(2)}K+\mathcal{P}_{\mathcal{G}}^{(1)}\mathcal{P}_{\mathcal{G}}^{(2)}K,
\label{defpk}
\end{eqnarray}
where $\mathcal{P}_{\mathcal{G}}^{(1)}$, $\mathcal{P}_{\mathcal{G}}^{(2)}$ are projection transformations from $L_2(\Omega\times \Omega) $ to $L_2(\Omega\times \Omega)$,
{ \begin{eqnarray}
 \begin{aligned}
\mathcal{P}_{\mathcal{G}}^{(1)}K(\boldsymbol x_1,\boldsymbol x_2)&=\sum_{j=1}^q e_j(\boldsymbol x_1)\int_{\Omega}K(\boldsymbol x,\boldsymbol x_2)e_j(\boldsymbol x)d \boldsymbol x,\\
\mathcal{P}_{\mathcal{G}}^{(2)}K(\boldsymbol x_1,\boldsymbol x_2)&=\sum_{j=1}^q e_j(\boldsymbol x_2)\int_{\Omega}K(\boldsymbol x_1,\boldsymbol x)e_j(\boldsymbol x)d \boldsymbol x,\\
\mathcal{P}_{\mathcal{G}}^{(1)}\mathcal{P}_{\mathcal{G}}^{(2)}K(\boldsymbol x_1,\boldsymbol x_2)
&=\sum_{i=1}^q \sum_{j=1}^q e_i(\boldsymbol x_1)e_j(\boldsymbol x_2)\int_{\Omega}\int_{\Omega} K(\boldsymbol x,\boldsymbol t)e_i(\boldsymbol x)e_j(\boldsymbol t)d \boldsymbol x d \boldsymbol t.
\end{aligned}
\end{eqnarray}}
Theorem 3.2 in \cite{tuo2019adjustments} proves the positive definiteness of $ K_{\mathcal{G}}$.
\subsection{Projected kernel calibration} 

{ Suppose $\bm\theta^*$ is an interior point of $\Theta$.
A necessary optimality condition for (\ref{thetastar})  to hold is}
 \begin{eqnarray}
\int_{\Omega}\frac{\partial{y^s(\bm x,\bm\theta^{*})}}{\partial{\theta_j}}\delta^*(\bm x)d\bm x=0,
\label{or cond}
\end{eqnarray}
for $j=1,2,\ldots,q$. 
{ For $\bm\theta=(\theta_1,\ldots,\theta_q) \in \Theta$, define 
 \begin{eqnarray}
 \mathcal{G}_{\theta}=span\left\{\frac{\partial y^s(\cdot,\bm\theta)}{\partial \theta_1},\frac{\partial y^s(\cdot,\bm\theta)}{\partial \theta_2},\ldots,\frac{\partial y^s(\cdot,\bm\theta)}{\partial \theta_q}\right\}. 
\label{or space}
\end{eqnarray}}
Let ${K_{\theta}}= K_{\mathcal{G}_{\theta}}$ to be the projected kernel and $\mathcal{N}_{K_{\theta}}(\Omega)$  the reproducing kernel Hilbert space  generated by $K_{\theta}$ \cite{wendland2004scattered}. 
{Then $\delta^*(\cdot)$ is in the orthogonal of $\mathcal{G}_{\theta^*}$ in light of (\ref{defpk}) and (\ref{or cond}).}
By assuming $\delta^*\in \mathcal{N}_{K_{\theta^*}}(\Omega)$,
the projected kernel smoothing estimator $(\hat{\bm\theta}_{PK},\hat\delta_{PK})$ is defined as the minimizer of
 { \begin{eqnarray}
\min_{\bm\theta\in \Theta}\min_{ \delta\in {\mathcal{N}_{K_{\theta}}(\Omega)}}\frac{1}{n}\sum_{i=1}^n\left(y_i-\delta(\bm x_i)-y^s(\bm x_i,\bm\theta)\right)^2+\lambda\|\delta\|^2_{\mathcal{N}_{K_{\theta}}(\Omega)},
\label{estimator}
\end{eqnarray}}
 where $\lambda$ is a tuning parameter, which can be choosen by generalized cross validation (GCV); see \cite{james2013bias}. 
  Following from the representer’s theorem \cite{scholkopf2001generalized, wahba1990spline}, $\hat{\bm\theta}_{PK}$ can also be represented by 
 \begin{eqnarray}
\hat{\bm\theta}_{PK}=\operatorname*{argmin}_{\Theta}(\bm Y-\bm Y^s_{\theta})^{T}(\mathbf K_{\theta}+n\lambda \mathbf{I}_n)^{-1}(\bm Y-\bm Y^s_{\theta}).
\label{krrcalibration}
\end{eqnarray}
where $\mathbf{I}_n$ is the identity matrix, $\bm{Y}_{\theta}^s=\left(y^s(\bm{x}_1,\bm\theta),\ldots,y^s(\bm{x}_{n},\bm\theta)\right)^{T}$, and $ \mathbf K_{\theta}=[K_{\theta}(\bm x_i,\bm x_j)]_{1\leq i,j\leq n}$.   
Theorem 4.3 of \cite{tuo2019adjustments} shows that the projected kernel estimation $\hat{\bm\theta}_{PK}$ is asymptotically normally distributed and there does not exist a regular estimator with even smaller asymptotic variance than $\hat{\bm\theta}_{PK}$.
 
 A bayesian interpretation for the projected kernel calibration is also given in \cite{tuo2019adjustments}.
  Suppose $\zeta(\cdot)-y^s(\cdot,\bm\theta)$ is a realization of Gaussian Process $GP(0, \tau^2 K_{\theta})$, where $\tau^2$ is the variance of the covariance function and  $K_{\theta}$ is the correlation function.   Denote the prior distribution of $\bm \theta$ as $\pi(\bm\theta)$, the posterior of $\bm\theta$ can be presented as
  \begin{eqnarray}
  \label{bay}
  \pi(\bm\theta|\bm{Y},\bm{Y}^s)\propto\pi(\bm\theta)\times \exp\left\{-\frac{1}{2}(\bm Y-\bm Y^s_{\theta})^{T}(\mathbf K_{\theta}+n\lambda \mathbf{I}_n)^{-1}(\bm Y-\bm Y^s_{\theta})\right\},
 \end{eqnarray}
 where $n\lambda=\sigma^2/\tau^2$ and  an uninformative prior is used in projected kernel calibration, that is {$\pi(\bm\theta)\propto1$.}

\section{Performance of the projected kernel calibration}
\label{comparison}
In this Section, we {examine} the performance of the projected kernel calibration. In subsection \ref{3.1}, we introduce the PK loss function and discuss the relationship between the local extrema of the $L_2$ loss function and those of the PK loss function. The uniform convergence of 
the PK loss function is proved in subsection  \ref{con.pkloss}, which shows that all 
the local minima of the PK loss function converge to a single value.  In subsection \ref{ex}, the numerical study  in \cite{gu2018scaled} is revisited to validate our theoretical findings.

 \subsection{Local minima of the PK loss function}
 \label{3.1}
The straightforward way to find the local minima of the PK loss function is to 
calculate its first and the second derivatives: the first derivative at any 
local minimum is zero and the Hessian matrix is {positive semi-definite}. To this aim,  
we first rewrite  (\ref{estimator}) to explicit the dependence on $\bm\theta$. 
  \begin{proposition}
  \label{th:equatheta} 
  Define $\delta^{\theta}(\cdot)=\zeta(\cdot)-y^s(\cdot,\bm\theta)$, $\delta_i^{\theta}=y_i-y^s(\bm x_i,\bm\theta)$. 
{  For each $\bm\theta\in \Theta$, let $\hat\delta^{\theta}_{PK}$ to be the  projected kernel smoothing estimator of $\delta^{\theta}$ which is defined as
  \begin{eqnarray}
  \label{deltapkk}
\hat\delta^{\theta}_{PK}=\operatorname*{argmin}_{\delta\in \mathcal{N}_{K_{\theta}}(\Omega)} \frac{1}{n}\sum_{i=1}^n\left(\delta_i^{\theta}-\delta(\bm x_i)\right)^2+\lambda\|\delta\|^2_{\mathcal{N}_{K_{\theta}}(\Omega)}.
\label{krr}
\end{eqnarray}}
 Then $\hat{\bm\theta}_{PK}$ can be written as
\begin{eqnarray}
\hat{\bm\theta}_{PK}=\operatorname*{argmin}_{\Theta}L_{PK}(\bm\theta),
\label{krrc}
\end{eqnarray}
where
  \begin{eqnarray}
L_{PK}(\bm\theta)=\frac{1}{n}\sum_{i=1}^n\left(\delta_i^{\theta}-\hat\delta_{PK}^{\theta}(\bm x_i)\right)^2+\lambda\|\hat\delta_{PK}^{\theta}\|^2_{\mathcal{N}_{K_{\theta}}(\Omega)}.
\label{pkloss}
\end{eqnarray}
The loss function (\ref{pkloss}) is referred to as the \emph{projected kernel loss function} (abbreviated as PK loss function). 

\end{proposition}

{ { From  (\ref{pkloss}),  we can see that} the derivative of the PK loss function depends on the derivatives of $\hat\delta_{PK}^{\theta}$ and $\|\hat\delta_{PK}^{\theta}\|^2_{\mathcal{N}_{K_{\theta}}(\Omega)}$ with respect to $\bm\theta$. 
Since $\hat\delta^{\theta}_{PK}$ belongs to the functional space ${\mathcal{N}_{K_{\theta}}(\Omega)}$, which itself depends on $\bm\theta$, it is difficult to evaluate the derivatives of $\hat\delta_{PK}^{\theta}$ and $\|\hat\delta_{PK}^{\theta}\|^2_{\mathcal{N}_{K_{\theta}}(\Omega)}$ with respect to $\bm\theta$ directly. 

{To solve this problem, we look for a function $\hat\delta^{\theta}$ in the functional space $\mathcal{N}_{K}(\Omega)$ such that $\hat\delta^{\theta}_{PK}=\mathcal{P}^{\perp}_{\mathcal{G}_{\theta}}\hat\delta^{\theta}=\hat\delta^{\theta}-\mathcal {P}_{\mathcal{G}_{\theta}}\hat\delta^{\theta}$. }{Here $\mathcal{N}_{K}(\Omega)$  denotes the reproducing kernel Hilbert space  generated by $K$. It is clear that  $\mathcal{N}_{K}(\Omega)$ is not dependent on $\bm\theta$.} 
{A natural choice of $\hat\delta^{\theta}$ is the summation of  $\hat\delta^{\theta}_{PK}$  (\ref{deltapkk}) with a certain function in ${\mathcal{G}_{\theta}}$ to be determined.
Combining the definition of $\hat\delta^{\theta}_{PK}$ with the generalized representation theorem  \cite{scholkopf2001generalized}, Proposition \ref{th:lossforderiv}  defines $\hat\delta^{\theta}$.}

 \begin{proposition}
  \label{th:lossforderiv} 
{   For each $\bm\theta\in \Theta$, define $\hat\delta^{\theta}$ to be the minimizer of the following loss function
\begin{eqnarray}
l_{\theta}(\delta_0)=\frac{1}{n}\sum_{i=1}^n\left(\delta_i^{\theta}-\mathcal{P}_{\mathcal{G}_{\theta}}\delta_0(\bm x_i)\right)^2+\frac{1}{n}\sum_{i=1}^n\left(\delta_i^{\theta}-\mathcal{P}^{\perp}_{\mathcal{G}_{\theta}}\delta_0(\bm x_i)\right)^2+\lambda\|\mathcal{P}^{\perp}_{\mathcal{G}_{\theta}}\delta_0\|^2_{\mathcal{N}_{K_{\theta}}(\Omega)},
\end{eqnarray}
where $\delta_0 \in \mathcal{N}_K(\Omega)$. 
Assuming $\mathcal{G}_{\theta}$ (\ref{or space}) is  a subspace of $\mathcal{N}_K(\Omega) $, 
 then 
$$\hat\delta^{\theta}_{PK}=\mathcal{P}^{\perp}_{\mathcal{G}_{\theta}}\hat\delta^{\theta}=\hat\delta^{\theta}-\mathcal {P}_{\mathcal{G}_{\theta}}\hat\delta^{\theta}.$$}
\end{proposition}

\begin{remark}From Theorem 3.3 in   \cite{tuo2019adjustments},  
{ $ \mathcal{N}_K(\Omega)$ is the direct sum of $ \mathcal{N}_{K_{\theta}}(\Omega)$ and $\mathcal{G}_{\theta}$.} Using this fact, 
one may easily prove that $\hat\delta^{\theta}_{PK}$  (\ref{deltapkk}) is a minimizer of the loss function $l_{\theta}$ in  $ \mathcal{N}_{K_{\theta}}(\Omega)$. 
\end{remark}
  
 {According to  Proposition   \ref{th:lossforderiv}, the calculation of the derivatives of $\hat\delta^{\theta}_{PK}$ is transformed into the calculation of the derivatives of $\hat\delta^{\theta}$ and $\mathcal {P}_{\mathcal{G}_{\theta}}\hat\delta^{\theta}$ on $\bm\theta$.
To establish the asymptotic behavior of the  derivatives of the PK loss function on $\bm\theta$, we first analyze how accurately  $\hat\delta^{\theta}$ may approximate  $\delta^{\theta}$ in a uniform sense. }

    \begin{proposition}
\label{condelta}
 Define the empirical norm of $g\in \mathcal{N}_{K_{\theta}}(\Omega)$ as $\|g\|^2_n=\frac{1}{n}\sum_{i=1}^n g^2(\bm x_i)$. Then under the following hypotheses

A1. $x_i$'s are independent random samples from the uniform distribution over $\Omega$.


A2. $\mathcal{N}_K(\Omega)$ can be continuously embedded into the (fractional) Sobolev space $H^m(\Omega)$ with $m>d/2$.

A3. It holds that
 $$\sup_{{\bm \theta\in\Theta},j=1,2,\ldots,q}\left\{\|\frac{\partial y^s(\cdot,\bm\theta)}{\partial \theta_j}\|_{\mathcal{N}_K(\Omega)}/\|\frac{\partial y^s(\cdot,\bm\theta)}{\partial \theta_j}\|_{L_2(\Omega)}\right\}<\infty,$$
  $$\sup_{{\bm \theta\in\Theta},j=1,2,\ldots,q}\left\{\|\frac{\partial y^s(\cdot,\bm\theta)}{\partial \theta_j}\|_{L_2(\Omega)}\right\}<\infty,$$
 and
 $$\sup_{\bm \theta\in\Theta}\|y^s(\cdot,\bm\theta)\|_{\mathcal{N}_K(\Omega)}<\infty.$$

 A4. Define the matrix $\mathbf E_{\theta}$ as
  \begin{eqnarray} 
  \label{etheta}
 \mathbf E_{\theta}=[\int \frac{\partial{y^s(\bm x,\bm \theta)}}{\partial{\theta_i}}\frac{\partial{y^s(\bm x,\bm \theta)}}{\partial{\theta_j}}d\bm x]_{1\leq i,j\leq q}.
  \end{eqnarray} 
Assume   $\inf_{\bm \theta\in\Theta} \lambda_{\min}(\mathbf E_{\theta})>0$, where $\lambda_{\min}(\mathbf E_{\theta})$ is the minimum eigenvalue of $\mathbf E_{\theta}$. 

 We have that if $\lambda\sim n^{-\frac{2m}{2m+d}}$, 
 \begin{eqnarray} 
    \begin{aligned}
   & \sup_{\boldsymbol \theta\in\Theta}\left\|{\delta}^{\theta}-\hat{\delta}^{\theta}\right\|_{L_2(\Omega)}=O_p(n^{-\frac{m}{2m+d}}),\\
& \sup_{\boldsymbol \theta\in\Theta}\| \hat \delta^{\theta}_{PK}\|_{\mathcal{N}_{K_{\theta}}(\Omega)}=O_p(1).
     \end{aligned}
      \end{eqnarray}

 \end{proposition}}

\begin{remark}
Condition A2 can be met easily if $K$ is chosen to be a Mat\'ern kernel (\ref{matern f}). By the Corollary 1 of \cite{tuo2016theoretical}, the reproducing kernel Hilbert space generated by the Mat\'ern kernel  is equal to the (fractional) Sobolev space $H^{\nu+d/2}(\Omega)$, with equivalent norms. 
\end{remark}

\begin{corollary}
\label{col3.6}
Under the conditions of Proposition \ref{condelta}, we have that  if $\lambda\sim n^{-\frac{2m}{2m+d}}$, then
 $$\sup_{\boldsymbol \theta\in\Theta}\left\|\mathcal{P}^{\perp}_{\mathcal{G}_{\theta}} {\delta}^{\theta}-\hat{\delta}^{\theta}_{PK}\right\|_{L_2(\Omega)}=O_p(n^{-\frac{m}{2m+d}}).$$
\end{corollary}

{ It is worth noticing that    Proposition \ref{condelta}  is different from  Theorem  4.1 in  \cite{tuo2019adjustments}. Let $\hat\zeta_{PK}(\cdot)=\hat\delta_{PK}^{\hat{ \theta}_{PK}}(\cdot)+y^s(\cdot,\hat{\bm \theta}_{PK})$ be the predictor given by the projected kernel calibration. Theorem  4.1 in  \cite{tuo2019adjustments}  shows that, under certain conditions, there is $\|\hat\zeta_{PK}-\zeta\|_{{L_2(\Omega)}}=O_p(n^{-\frac{m}{2m+d}})$. It is  an immediate consequence of the combination of Corollary \ref{col3.6} and  the asymptotic normal of the projected kernel calibration. Let us briefly explain the reasons. 

By the triangle inequality, $ \|\hat\zeta_{PK}-\zeta\|_{{L_2(\Omega)}}$ can be { bounded from above }
 \begin{eqnarray*}
 \begin{aligned}
&\left\|\hat\delta_{PK}^{\theta^*}-\delta^*\right\|_{{L_2(\Omega)}}+ \left\{\left \| \hat\delta_{PK}^{\hat{ \theta}_{PK}}-\hat\delta_{PK}^{\theta^*}\right\|_{{L_2(\Omega)}}+\left\|y^s(\cdot,\hat{\bm \theta}_{PK})-y^s(\cdot,\bm\theta^*)\right\|_{{L_2(\Omega)}}\right\}.
 \end{aligned}
  \end{eqnarray*}
Next, we bound these two terms separately. Because $\delta^*= {\delta}^{\theta^*}\in{\mathcal{G}^{\perp}_{\theta^*}}$, from Corollary 3.6, the fist term is equal to $O_p(n^{-\frac{m}{2m+d}})$. 
 It follows from Taylor's theorem, Cauchy-Schwarz inequality, and the asymptotic normal of the projected kernel calibration that,  the second term is equal to 
$O_p(n^{-1/2})$. The desired result of  Theorem  4.1 in  \cite{tuo2019adjustments} can be easily obtained by the  summation of these two bounds. }
 
       \begin{theorem}
  \label{th:stationary points} 
Under the conditions of Proposition \ref{condelta},  we have   that  
{ \begin{eqnarray} 
  \label{un-con}
\sup_{\bm\theta\in\Theta} \left|\frac{\partial L_{PK}(\bm\theta)}{\partial {\theta_j}} - \frac{\partial {\bm a^T_{\theta}\mathbf E^{-1}_{\theta} \bm a_{\theta}}}{\partial \theta_j}\right|=O_p(n^{-\frac{m}{2m+d}}), j=1,2,\ldots,q,
   \end{eqnarray} }
{ where  $\bm a^{T}_{\theta}=<\delta^{\theta}(\cdot),\frac{\partial y^s(\cdot,\bm\theta)}{\partial\bm\theta}>_{L_2(\Omega)}$}.

Let $\bm\theta^s$ to be a stationary point of the $L_2$ loss function (\ref{l2loss}),  then the Hessian matrix of $L_{PK}(\bm\theta)$ at $\bm\theta^s$, denotes as $\mathbf H_{PK}(\bm\theta^s)$ may be written as
 \begin{eqnarray} 
    \begin{aligned}
\mathbf H_{PK}(\bm\theta^s)=2\frac{\partial {\boldsymbol a^T_{\boldsymbol\theta^s}}}{\partial{\bm \theta}}\mathbf E^{-1}_{\boldsymbol\theta^s}\frac{\partial {\boldsymbol a^T_{\boldsymbol\theta^s}}}{\partial{\bm \theta}}+O_p(n^{-\frac{m}{2m+d}}).
    \end{aligned}
    \end{eqnarray}
\end{theorem}

\begin{corollary}
\label{cor8}
Under the condition of Proposition \ref{condelta}, we have that if $\bm\theta^s$ is a local maxima or local minima of { the $L_2$ loss function (\ref{l2loss})}, then $\lim_{n\rightarrow \infty} \frac{\partial L_{PK}(\bm\theta)}{\partial \bm \theta}|_{\bm\theta^s}=\bm 0$, i.e. $\bm\theta^s$ is a stationary point of the PK loss function. 
 \end {corollary}
 
Corollary \ref{cor8} is an immediate consequence of  (\ref{un-con}), following from  
the vanishing of the first derivative of $L_2(\bm\theta)$ at $\bm\theta^s$ is zero, which itself is equivalent to $\bm a_{\bm\theta^s}=\bm {0}$.  To check whether $\bm\theta^s$ is a local maximum or a local minimum of the PK loss function, let us study the sign of the Hessian matrix  $\mathbf H_{PK}(\bm\theta^s)$.
Since $\mathbf E^{-1}_{\boldsymbol\theta^s}$ is a positive definite matrix, then the matrix $ \frac{\partial {\boldsymbol a^T_{\boldsymbol\theta^s}}}{\partial{\bm \theta}}\mathbf E^{-1}_{\boldsymbol\theta^s}\frac{\partial {\boldsymbol a^T_{\boldsymbol\theta^s}}}{\partial{\bm \theta}}$ is positive semi-definite. When the sample size $n$ tends to infinity, we can easily prove the following corollary.
 \begin{corollary}
  \label{th:localmin} 
  Under the condition of Proposition \ref{condelta}, we have that if $\bm\theta^s$ is a local maxima or local minima of the $L_2$ loss function (\ref{l2loss}), then
the Hessian matrix of the PK loss at $\bm\theta^s$ is positive semi-definite. 
  \end{corollary}

 Corollary \ref{th:localmin} shows that the local optimal points of  the $L_2$ loss function are local minima of  the PK loss function. 
Specifically, suppose that the $L_2$ loss function has $l$ local optima, denoted as $\{\bm\theta^s_1,\ldots,\bm\theta^s_l\}$. If $l=1$,  then $\bm\theta^s_1$ is the global minimum of $L_{PK}$. If $l\geq 2$,  then $\bm\theta^s_1,\ldots,\bm\theta^s_l$ are local minima of $L_{PK}$.

    \subsection{Convergence of the PK loss function}
  \label{con.pkloss}
{ From Theorem 10.7 in \cite{carothers2000real}, we have that Theorem   \ref{th:stationary points}  alone does not indicate uniform convergence of the PK loss function. Another necessary condition is that the PK loss function converges at at least one point.
  Based on the definition of $L_{PK}(\bm\theta)$ (\ref{pkloss}), 
     \begin{eqnarray}
     \begin{aligned}
     \label{lpkstar}
  L_{PK}(\bm\theta^*)
  =&\frac{1}{n}\sum_{i=1}^n\left(\delta^{*}(\bm x_i)-\hat\delta_{PK}^{\theta^*}(\bm x_i)\right)^2+\frac{2}{n}\sum_{i=1}^n\epsilon_i\left(\delta^{*}(\bm x_i)-\hat\delta_{PK}^{\theta^*}(\bm x_i)\right)+\\
 &\frac{1}{n}\sum_{i=1}^n\epsilon^2_i+\lambda\|\hat\delta_{PK}^{\theta^*}\|^2_{\mathcal{N}_{K_{\theta^*}}(\Omega)}.
      \end{aligned}
       \end{eqnarray}
Combining  (\ref{lpkstar}),  (\ref{semi2}), (\ref{ep}) and Proposition \ref{condelta}, we have that if $\lambda\sim n^{-\frac{2m}{2m+d}}$, there is
  $L_{PK}(\bm\theta^*)=\frac{1}{n}\sum_{i=1}^n\epsilon_i^2+O_p(n^{-\frac{2m}{2m+d}})<\infty.$} Let us now denote $C_1=\lim_{n \rightarrow \infty}L_{PK}(\bm\theta^*)=\sigma^2<\infty$. The uniform convergence of $\frac{\partial L_{PK}(\bm\theta)}{\partial {\theta_j}},j=1,\ldots,q$ and the convergence of $L_{PK}(\bm\theta^*)$ guarantee the uniform convergence of $L_{PK}(\bm\theta)$, that is
{\begin{eqnarray} 
   \label{con.delta.1}
\lim_{n\rightarrow\infty}\sup_{\bm\theta\in\Theta} \left |{ L_{PK}(\bm\theta)} - {\bm a^T_{\theta}\mathbf E^{-1}_{\theta} \bm a_{\theta}}-C_1\right |=0.
   \end{eqnarray} }
Since $\mathcal{P}_{\mathcal{G}_{\theta}}{\delta^{\theta}}=\bm a^{T}_{\theta}\mathbf E^{-1}_{\theta} \frac{\partial y^s(\cdot,\bm\theta)}{\partial \bm\theta^T} $, it follows from basic linear algebra that
   \begin{eqnarray} 
   \label{pkdelta}
   \bm a^T_{\theta}\mathbf E^{-1}_{\theta} \bm a_{\theta}=<\delta^{\theta},\mathcal{P}_{\mathcal{G}_{\theta}}{\delta^{\theta}}>_{L_2(\Omega)}=\|\mathcal{P}_{\mathcal{G}_{\theta}}{\delta^{\theta}}\|^2_{L_2(\Omega)}.
     \end{eqnarray} 
Upon combining  (\ref{con.delta.1}) and  (\ref{pkdelta}), {Theorem \ref{pro.uniform.con}}  shows the uniform convergence of $L_{PK}(\bm\theta)$.

  \begin{theorem}
    \label{pro.uniform.con} 
 Under the condition of Proposition \ref{condelta},  we have 
  that 
{ \begin{eqnarray} 
   \label{con.delta.2}
\lim_{n\rightarrow \infty}\sup_{\bm\theta\in\Theta} \left |{ L_{PK}(\bm\theta)} - \|\mathcal{P}_{\mathcal{G}_{\theta}}{\delta^{\theta}}\|^2_{L_2(\Omega)}-C_1\right |=0.
   \end{eqnarray} }
      $\|\mathcal{P}_{\mathcal{G}_{\theta}}{\delta^{\theta}}\|^2_{L_2(\Omega)}+C_1$ is termed the \emph{projected kernel $L_2$ loss function} (PKL2 loss function).

   \end{theorem}

Since $\delta^{\theta^s_i}$ is orthogonal to $\mathcal{G}_{\theta^s_i}$, Theorem \ref{pro.uniform.con} shows that $\lim _{n\rightarrow \infty}L_{PK}(\bm\theta^s_i)=C_1$ for any $i=1,\ldots, l$.  
  That is, all the local minima of the PK loss functions approach the same value. 
 This implies that the projected kernel calibration method may easily get stuck in local minima or even local maxima of  the $L_2$ loss function.     

  \subsection{Revisit the Example 2 in  \cite{gu2018scaled}}
\label{ex}
To validate our theoretical assertions, the numerical study in  \cite{gu2018scaled} is revisited in this subsection. Example 2 in  \cite{gu2018scaled} shows that $L_{PK}$ tends to have more local optimal points than the $L_2$ loss function. In this example, the smoothness of discrepancy function is infinite, whereas \cite{gu2018scaled} choose a  Mat\'ern kernel function (\ref{matern f}) with $\nu=1/2$ to estimate the discrepancy function.  To avoid the effect of inaccuracy of the correlation function, we make some modifications to this example. 

Suppose the true process is $$\zeta(x)=x \cos(3x/2)+x, x\in [0,5].$$ 
 The computer model is $$y^s(x,\theta)=\sin(\theta x)+\exp(-2|x|), \theta\in [0,3].$$
By the definition of $\theta^*$ (\ref{thetastar}), {we find that } $\theta^*= 0.371$. Because the smoothness of $\zeta(\cdot)-y^s(\cdot,\theta)$ depends on the smoothness of $\exp(-2|x|)$, we have that the discrepancy function is in the reproducing kernel Hilbert space generated by the Mat\'ern kernel function $K(h;\frac{1}{2},\frac{1}{2})$, which is equal to the Sobolev space $H^{1}(\Omega)$. 
 
 Let $\{x_1, x_2,\ldots, x_{n}\}$ be the set of the design points, where $x_i=-5+\frac{10(i-1)}{n-1}, i=1,2,\ldots,n$ and $n=100$ is the sample size. Suppose the observation error $\epsilon_i$'s are mutually independent and follow from $N(0,0.2^2)$. 
 By Proposition \ref{condelta}, set the tuning parameter $\lambda=\beta n^{-\frac{2}{3}}$ and $\beta$ is chosen by $10$-fold cross validation method \cite{james2013bias}. With the help of \emph{caret }Package \cite{kuhn2015short} in R, we have that $\beta=0.00138$. To illustrate the performance of the projected kernel calibration more clearly, we scale these two loss function to $[0,1]$ respectively. The scaled value of loss function $L$ is defined as
 \begin{eqnarray}
 \frac{L(\bm\theta)-\min_{\Theta}L(\bm\theta)}{\max_{\Theta}L(\bm\theta)-\min_{\Theta}L(\bm\theta)}.
\end{eqnarray}

\begin{figure}[htbp]
  \centering
  \includegraphics[scale=0.65]{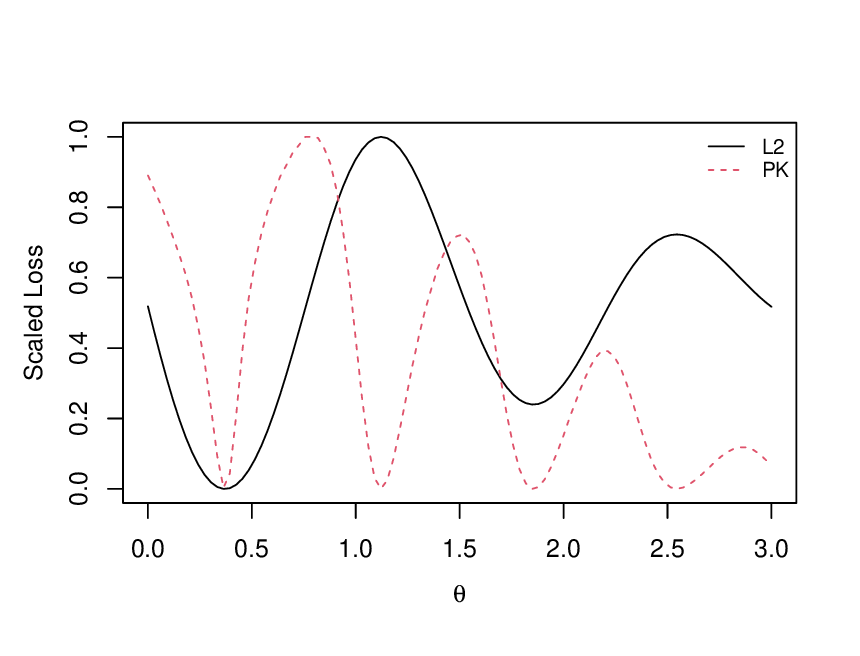}
  \caption{PK loss function (red dashed line) v.s. $L_2$ loss function (\ref{l2loss}) (black real line).}
  \label{fig:compare}
\end{figure}

 Figure \ref{fig:compare} shows the comparison between the PK loss function and the $L_2$ loss function when the sample size $n=100$. The scaled $L_2$ loss over $ \theta\in [0, 3]$ is shown in the black solid line, which contains a global minimum at $0.371$, a local minimum at $1.855$, two local maxima at $1.122$ and $2.545$ respectively. The PK loss function (dashed line) has a global minimum at $0.371$, and three local minima at the three local optimal points of the $L_2$ loss function.

\begin{figure}[htbp]
  \centering
  \includegraphics[scale=0.65]{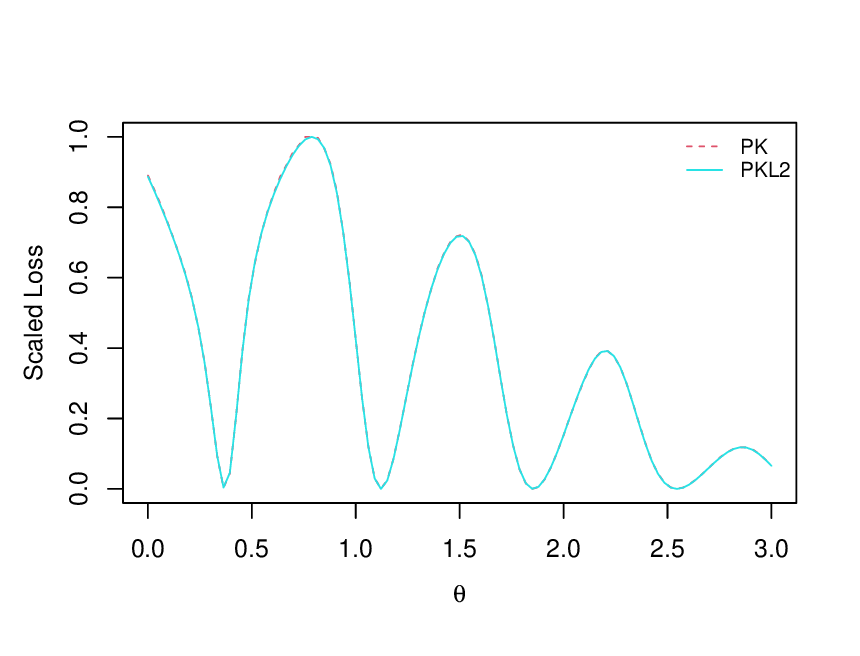}
  \caption{PK loss function (red dashed line) v.s. PKL2 loss function (blue real line).}
  \label{fig:conpk}
\end{figure}
 Figure \ref{fig:conpk} compares the PK loss function with the PKL2 loss function when the sample size $n=100$. We can see that the red dashed line is completely coincident with the green real line.  It indicates the PK loss function converges to the PKL2 loss function when the sample size is $100$.

\section{Penalized projected kernel calibration}
\label{sec4}
The uniform convergence of the PK loss function leads to the failure of any kind of global optimization. To overcome the problem, and inspired by \cite{gu2018scaled}, we rescale 
the $L_2$ norm of the discrepancy function, and introduce a penalized projected kernel calibration method. 
\subsection{Methodology}
We define the penalized projected kernel estimator of $\bm\theta$ as
 \begin{eqnarray}
\hat{\bm\theta}_{PPK}=\operatorname*{argmin}_{\Theta} \left\{L_{PK}(\bm\theta)+\eta\|\delta^{\theta}\|^2_{L_2(\Omega)}\right\},
\label{bpkloss2}
\end{eqnarray}
{ where $\eta>0$ is a tuning parameter to balance the PK loss and the $L_2$ loss.}
As $\delta^{\theta}$ is unknown, a natural choice is to replace $\delta^{\theta}$ by its estimator. 
Since Proposition \ref{condelta} guarantees that a proper choice to estimate $\delta^{*}$
is $\hat\delta_{PK}^{\theta^*}$, we propose to 
evaluate $\hat{\bm\theta}_{PPK}$ by  minimizing the following loss function
 \begin{eqnarray} 
    \begin{aligned}
L_{PPK}(\bm\theta)=L_{PK}(\bm\theta)+\eta\|\hat\delta_{PK}^{\theta}\|^2_{L_2(\Omega)}.
    \end{aligned}
    \label{ppkloss}
    \end{eqnarray}
We refer to the loss function (\ref{ppkloss}) as the \emph{penalized projected 
kernel loss function} (abbreviated as PPK loss function).  
By some direct calculations, we have that
 \begin{eqnarray}
 \label{ppkestimator}
\hat{\bm\theta}_{PPK}=\operatorname*{argmin}_{\Theta}(\bm Y-\bm Y^s_{\theta})^{T}\Sigma^{-1}(\bm Y-\bm Y^s_{\theta}),
 \end{eqnarray}
where  $\Sigma^{-1}=\left[(\mathbf K_{\theta}+n\lambda \mathbf{I}_n)^{-1}+\gamma(\mathbf K_{\theta}+n\lambda \mathbf{I}_n)^{-1}\|K_{\theta}(\bm x,\bm X)\|^2_{L_2(\Omega)}(\mathbf K_{\theta}+n\lambda \mathbf{I}_n)^{-1}\right]$ and $\gamma=\eta/\lambda$. 
This expression provides a natural Bayesian interpretation of the 
penalized projected kernel calibration, i.e.
 \begin{eqnarray}
 \label{bayes}
  \pi(\bm\theta|\bm{Y},\bm{Y}^s)\propto\pi(\bm\theta)\times \exp\left\{-\frac{1}{2}(\bm Y-\bm Y^s_{\theta})^{T}(\mathbf K_{\theta}+n\lambda \mathbf{I}_n)^{-1}(\bm Y-\bm Y^s_{\theta})\right\},
 \end{eqnarray}
where 
\begin{eqnarray}
\pi(\bm\theta)\propto\exp\left\{-\gamma/2\int_{\Omega}\left(\hat\delta^{\theta}_{PK}(\bm x)\right)^2d\bm x\right\}.
\end{eqnarray}
We can easily make a comparison between the projected kernel calibration and the proposed calibration from their Bayesian interpretations  (\ref{bay}) and (\ref{bayes}). Given the definition of $\bm\theta^*$, the estimation of $\bm\theta$ favors to values where $\|\hat\delta^{\theta}_{PK}\|^2_{L_2(\Omega)}$ is small. In turn, 
the prior distirbution $\pi(\bm\theta)$ of the proposed method is inversely 
proportional to $\|\hat \delta^{\theta}_{PK}\|^2_{L_2(\Omega)}$, which appears more suitable than the uninformative prior used in (\ref{bay}).

\subsection{Asymptotic properties}
In this section, we investigate the asymptotic properties of the proposed estimator. 
We first address the number of local minima of the PPK loss function, and then show that under certain conditions the proposed estimator of $\bm\theta$ is semi-parametrically efficient. Finally, we assess the predictive power of the proposed 
method in estimating the true process $\zeta(\cdot)$.

 \begin{theorem}
 \label{th:stationary points ppk} 
 Under the conditions of Proposition \ref{condelta}, suppose there exist constants $U\geq L$, such that
$$U\mathbf E_{\theta^s}\geq -\frac{\partial^2 L_2(\bm\theta)}{\partial\bm\theta\partial\bm\theta^T}|_{\bm\theta=\bm\theta^s}\geq L\mathbf E_{\theta^s}.$$
{If $\eta\in\Gamma_{\eta}$, where $\Gamma_{\eta}\subset (0,\infty)$ is an interval determined by $U$ and $L$ and 
the specific form of  $\Gamma_{\eta}$   is given in (\ref{gamma}),
we have that asymptotically (i.e. for the sample size $n\rightarrow \infty$) $\bm\theta^s$ is a local minimum (maximum) of PPK loss function if  $\bm\theta^s$ is a local minimum (maximum) of the $L_2$ loss function.}
\end{theorem}

The theorem means that by choosing an appropriate value of $\eta$, we may avoid 
the problem of having too many local minima. Next, we turn to the asymptotic 
properties of  $\hat {\boldsymbol\theta}_{PPK}$. 

\begin{theorem}
\label{consistppk}
In addition to the assumptions of Proposition \ref{condelta}, assume that
  

B1.  The matrix
$$\mathbf{V}=\int_{\Omega}\frac{\partial^2}{\partial {\bm\theta}^T\partial{\bm\theta}}\left(\zeta(\bm x)-y^s(\bm x,\bm\theta^*)\right)^2 d\bm x$$
is positive definite.

B2.  There exists a neighborhood of $\bm\theta^*$, denoted as $\Theta'$, satisfying
 $$\sup_{{\bm \theta\in\Theta'},j=1,2,\ldots,q}\left\{\|\frac{\partial y^s(\cdot,\bm\theta)}{\partial \theta_j}\|_{\mathcal{N}_{K_{\theta}}(\Omega)}\right\}<\infty,$$
   $$\sup_{{\bm \theta\in\Theta'},1\leq i,j\leq q}\left\{\|\frac{\partial^2 y^s(\cdot,\bm\theta)}{\partial \theta_i\partial \theta_j}\|_{\mathcal{N}_{K_{\theta}}(\Omega)}\right\}<\infty.$$
 Then we have 
\begin{eqnarray} 
\hat {\boldsymbol\theta}_{PPK}-\boldsymbol\theta^*=2\mathbf{V}^{-1}\left\{\frac{1}{n}\sum_{i=1}^n \epsilon_i\frac{\partial y^s(\boldsymbol x_i,\boldsymbol\theta^*)}{\partial \boldsymbol\theta}\right\}+o_p(n^{-1/2}).
\end{eqnarray} 
\end{theorem}

Theorem \ref{consistppk} shows the asymptotic normality of $\hat {\boldsymbol\theta}_{PPK}$:
$$\sqrt{n} (\hat {\boldsymbol\theta}_{PPK}-\boldsymbol\theta^*)\sim N(\bm 0, 4\sigma^2\mathbf{V}^{-1}\mathbf E_{\theta^*}\mathbf{V}^{-1}).$$
It is worth noticing that, the asymptotic representation of  $\hat {\boldsymbol\theta}_{PPK}-\boldsymbol\theta^*$ agrees with  Theorem 4.3 of \cite{tuo2019adjustments}. { It also shows that} the penalized projected kernel calibration is semi-parametrically efficient.

Let $\hat\zeta_{n}(\cdot)=\hat\delta^{\hat{\bm\theta}_{PPK}}_{PK}(\cdot)+y^s(\cdot,\hat{\bm\theta}_{PPK})$, then $\hat\zeta_{n}$ is a natural estimator of  $\zeta$. Theorem \ref{predpower} gives the predictive power of the proposed method.
\begin{theorem}
\label{predpower}
Under the conditions of Theorem \ref{consistppk}, we have 
\begin{eqnarray} 
\label{rateofzeta}
\|\hat\zeta_n-\zeta\|_{L_2(\Omega)}=O_p(n^{-\frac{m}{2m+d}}).
\end{eqnarray} 
\end{theorem}

The rate of convergence in (\ref{rateofzeta}) equals  the minimax rate in the current context \cite{stone1982optimal}.

\section{Addressing computational problems}
\label{comp}
Evaluating $ \hat {\boldsymbol\theta}_{PPK}$ has two major difficulties in practice. The first
problem is the calculations of projected kernels, since it is hard to evaluate $K_{\theta}$ from its definition  (\ref{defpk}).  The second problem is the  choice of $\eta$. We focus on these problems in this section.

\subsection{Calculus for projected kernels}
Let $\bm g^{T}_{\theta}(\cdot)=\frac{\partial y^s(\cdot,\bm\theta)}{\partial\bm\theta}$ and $\bm h_{\theta}(\bm x)=<K(\bm x,\cdot),\bm g_{\theta}(\cdot)>_{L_2(\Omega)}$.  A closed form for $K_{\theta}$ is derived in this subsection. Because $\mathbf E_{\theta}$ is positive definite, it follows from basic linear algebra that
 \begin{eqnarray}
 \begin{aligned}
 \mathcal{P}^{(1)}_{\mathcal{G}_{\theta}}K(\bm x_1,\bm x_2)= \mathcal{P}^{(2)}_{\mathcal{G}_{\theta}}K(\bm x_2,\bm x_1)&=\bm h^{T}_{\theta}(\bm x_2)\mathbf E^{-1}_{\theta} \bm g_{\theta}(\bm x_1),\\
  \mathcal{P}^{(1)}_{\mathcal{G}_{\theta}}\mathcal{P}^{(2)}_{\mathcal{G}_{\theta}}K(\bm x_1,\bm x_2)&= \bm g^T_{\theta}(\bm x_2)\mathbf E^{-1}_{\theta} \mathbf H_{\theta}\mathbf E^{-1}_{\theta} \bm g_{\theta}(\bm x_1),
 \end{aligned}
 \end{eqnarray}
 where $\mathbf H_{\theta}=\int \int K(\bm t_1,\bm t_2)  \bm g_{\theta}(\bm t_1) \bm g^{T}_{\theta}(\bm t_2) d \bm t_1 \bm t_2$.
Let $\bm w_{\theta}(\bm x)=\mathbf H_{\theta}\mathbf E^{-1}_{\theta} \bm g_{\theta}(\bm x)-\bm h_{\theta}(\bm x)$, then $K_{\theta}(\bm x_1,\bm x_2)$  (\ref{defpk}) can be represented as 
  \begin{eqnarray}
  \label{prokernels}
 \begin{aligned}
&K(\bm x_1,\bm x_2)+\bm w_{\theta}(\bm x_1)^{T}\mathbf H^{-1}_{\theta}\bm w_{\theta}(\bm x_2)-\bm h^{T}_{\theta}(\bm x_1)\mathbf H^{-1}_{\theta}\bm h_{\theta}(\bm x_2).
  \end{aligned}
 \end{eqnarray}
  
 Tuo \cite{tuo2019adjustments}  points out that, projected kernel calibration is similar to the Bayesian calibration method proposed by \cite{plumlee2017bayesian}, which is based on an orthogonal Gaussian process (OGP) modeling technique. {The covariance function of an orthogonal Gaussian process which is defined as   
   \begin{eqnarray}
   \label{kor}
K_{or}(\bm x_1,\bm x_2)=K(\bm x_1,\bm x_2)-\bm h^{T}_{\theta}(\bm x_1)\mathbf H^{-1}_{\theta}\bm h_{\theta}(\bm x_2)
 \end{eqnarray}
  is a projected kernel function.}
By comparing  (\ref{prokernels}) with (\ref{kor}), we have that, if and only if $\bm w_{\theta}(\bm x)=\bm 0$, there is $K_{\theta}=K_{or}$. To address the difficult integrations, we refer to \cite{plumlee2017bayesian} and approximate $<f_1,f_2>_{L_2(\Omega)}$ by
  \begin{eqnarray}
  \label{approx}
  \frac{1}{N}\sum_{k=1}^N f_1(\xi_k)f_2(\xi_k),
 \end{eqnarray}
where $\xi_k$'s are independent random samples from the uniform distribution over $\Omega$. 
By the strong law of large numbers,  (\ref{approx}) almost surely converges to $<f_1,f_2>_{L_2(\Omega)}$ as
$N\rightarrow \infty$. Through this approximation, $\bm h_{\theta}(\bm x)$, $\mathbf E^{-1}_{\theta}$ and $\mathbf H_{\theta}$ can be represented as
  \begin{eqnarray}
  \begin{aligned}
  \bm h_{\theta}(\bm x)&= \frac{1}{N}\sum_{k=1}^N K(\bm x,\xi_k)\bm g_{\theta}(\xi_k),\\
   \mathbf E_{\theta}&=\frac{1}{N}\sum_{k=1}^N \bm g_{\theta}(\xi_k) \bm g^{T}_{\theta}(\xi_k),\\
   \mathbf H_{\theta}&=\frac{1}{N^2}\sum_{i=1}^N\sum_{j=1}^N K(\xi_i,\xi_j)  \bm g_{\theta}(\xi_i) \bm g^{T}_{\theta}(\xi_j). 
    \end{aligned}
  \end{eqnarray}

    \subsection{Choice of $\eta$}
 The choice of the tuning parameter $\eta$ affects the number of local optima of $L_{PPK}(\bm\theta)$. In particular, by increasing $\eta$ from $0$ to $\infty$, one may gradually turn $L_{PPK}(\bm\theta)$  from rough to smooth. In this subsection, a BIC-like 
criterion is introduced to choose $\eta$. 

    Let ${\rm RI}(L)$ be an indicator that measure the ruggedness of a loss function $L$, such as the number of local optimal points. This indicator satisfies that:
  \begin{itemize}
\item ${\rm RI}(L)\geq 0$,  and the equality holds if and only if the loss function $L$ is a nonnegative constant;
\item ${\rm RI}(C_2L)={\rm RI}(L)$, where $C_2>0$ is a constant; 
\item ${\rm RI}(L+C_3)={\rm RI}(L)$, where $C_3>0$ is a constant;   
\item If ${\rm RI}(L_1)\leq {\rm RI}(L_2)$ then ${\rm RI}(L_1)\leq {\rm RI} (L_1+L_2)\leq {\rm RI}(L_2)$. 
\end{itemize}

 Theorem \ref{th:localmin} says that when the $L_2$ loss function has more than one local extrema, $L_{PK}(\bm\theta)$ tends to have more local extrema than $ \|\hat\delta_{PK}^{\theta}\|^2_{L_2(\Omega)}$. Therefore,  it is easy to see that, ${\rm RI}(L_{PK}(\bm\theta))\geq {\rm RI}(L_{PPK}(\bm\theta))\geq{\rm RI}( \|\hat\delta_{PK}^{\theta}\|^2_{L_2(\Omega)})$, with $ {\rm RI}(L_{PPK})$ being a decreasing function of $\eta$.

We may thus use a BIC-like criterion to estimate $\eta$ as follows
 \begin{eqnarray}
 \label{bic}
 \eta=\operatorname*{argmin} _{\eta}\log(L_{PK}({\hat{\bm \theta}}_{PPK}))+{\rm RI}(L_{PPK})\log(n)/n.
\end{eqnarray}

 In the field of optimization, there are many indicators that  may be used to measure the smoothness of the objective function  \cite{talbi2009metaheuristics}.  
A natural choice among these indicators is the number of local optima of the objective function. Upon denoting the number of local optima of the loss function $L$ by 
$NLO(L)$, we have that
  \begin{eqnarray}
 NLO(L)=\#\{\bm\theta: \frac{\partial L(\bm\theta)}{\partial\bm\theta^T}=\bm 0\}.
  \end{eqnarray} 
In most cases, one cannot find a closed form for $\frac{\partial L(\bm\theta)}{\partial\bm\theta^T}$. By numerically approximating \cite{stoer2013introduction}  the first derivative of $L$ on $\bm\theta$, one may use Newton-Raphson method \cite{soetaert2009rootsolve} to evaluate the number of local optima of the loss function $L$. If the tuning parameter $\eta$ is chosen according to the $NLO$ index, $L_{PPK}$ is termed as PPK.NLO loss function.
  
 Notice that when the dimension of $\bm\theta$ is large, it is always hard to count the number of local extrema of a loss function. In turn, amplitude {indices} which measure the distribution of local minima of the loss function $L$, are widely used to assess the smoothness of a function \cite{talbi2009metaheuristics}. We employ the following definition 
 \begin{eqnarray}
Amp(L)=\frac{\max L(\bm\theta)-\min L(\bm\theta)}{\int_{\theta}L(\bm\theta)-\min L(\bm\theta)d\bm\theta}.
\end{eqnarray}
The larger is $Amp(L)$, the harder is to find the optimal point for the loss function $L$.  The PPK loss function where $\eta$ is chosen by the $Amp$ index is referred to as 
PPK.Amp loss function.
Let us denote by $\Theta_s=\{\bm\theta_1,\ldots,\bm\theta_{N'}\}$ a discrete set 
of values of $\bm\theta$, where $\bm\theta_k, k=1,\ldots,N'$ is randomly sampled from the uniform distribution over $\Theta$. We approximate $Amp(L)$ by
 \begin{eqnarray}
 \label{disamp}
\frac{\max_{\Theta_s} L(\bm\theta)-\min_{\Theta_s} L(\bm\theta)}{\frac{1}{N'}\sum_{i=1}^{N'} \left[L(\bm\theta)-\min_{\Theta_s} L(\bm\theta)\right]}.
\end{eqnarray}
 
\section{Numerical studies}
\label{sec5}
In this section, we examine the performance of the penalized projected kernel calibration by using two simulated examples and one real case study.   
In subsection \ref{ex1}, we go back to the example already discussed in \ref{ex}, whereas in subsection \ref{ex2}, we study a simulated example with a two-dimensional calibration parameter. 
We compare the performance of the penalized projected kernel calibration with some commonly used calibration methods using samples with sizes.  To ensure a fair comparison, the mean
function, correlation function, as well as the $\{\xi_j\}$'s in the integration 
are the same.
 \subsection{Review of example \ref{ex}}
 \label{ex1}

To assess the performance of the proposed method, we compare the PPK loss function 
with the $L_2$ loss function and the PK loss function when the sample sizes are $n=\{6,15,100\}$. The physical design and kernel function are the same as in \ref{ex}. The tuning parameter $\eta$ in the  PPK loss function is chosen by the BIC-like criterion (\ref{bic}). We use the two quantities $NLO(L)$ and $Amp(L)$ to quantify the smoothness 
of the loss function. 
 \begin{figure}[htbp]
 \centering
 \subfigure[$\zeta( x)$ v.s. $y^s(x,\theta^*)$.]{
\begin{minipage}[t]{0.45\linewidth}
\centering
\includegraphics[width=2.7in]{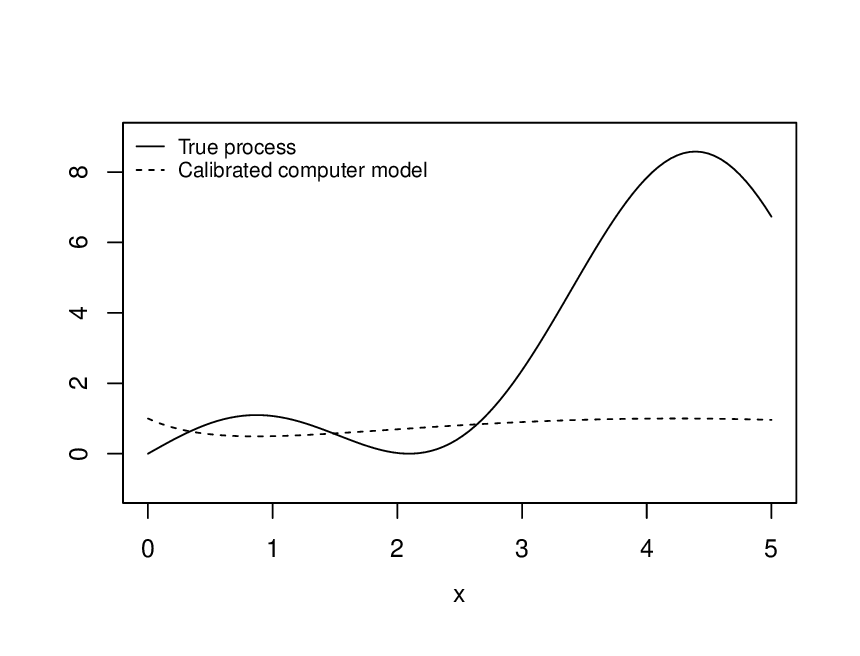}
\end{minipage}%
}%
\centering
\subfigure[$n=6$.]{
\begin{minipage}[t]{0.45\linewidth}
\centering
\includegraphics[width=2.7in]{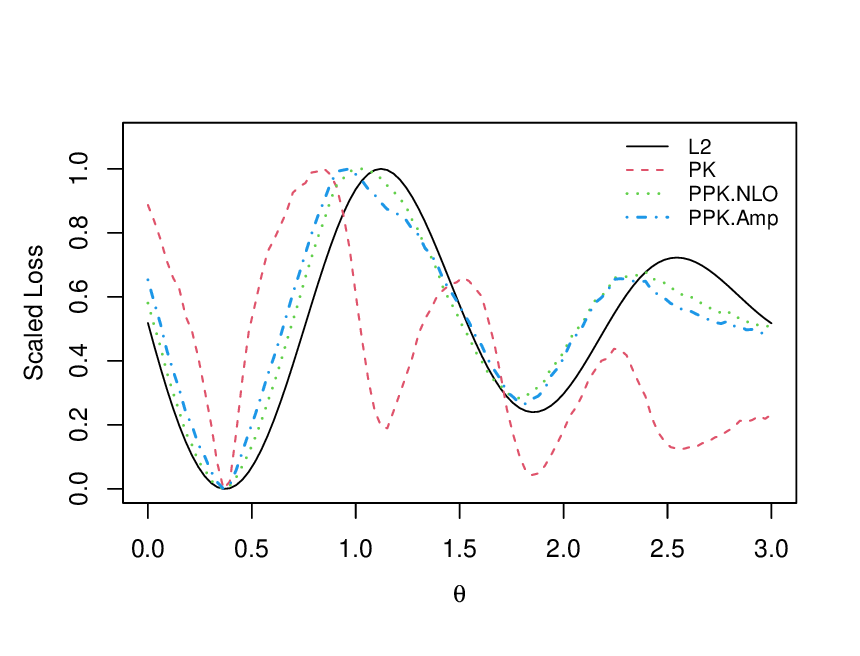}
\end{minipage}%
}%

\subfigure[$n=15$.]{
\begin{minipage}[t]{0.45\linewidth}
\centering
\includegraphics[width=2.7in]{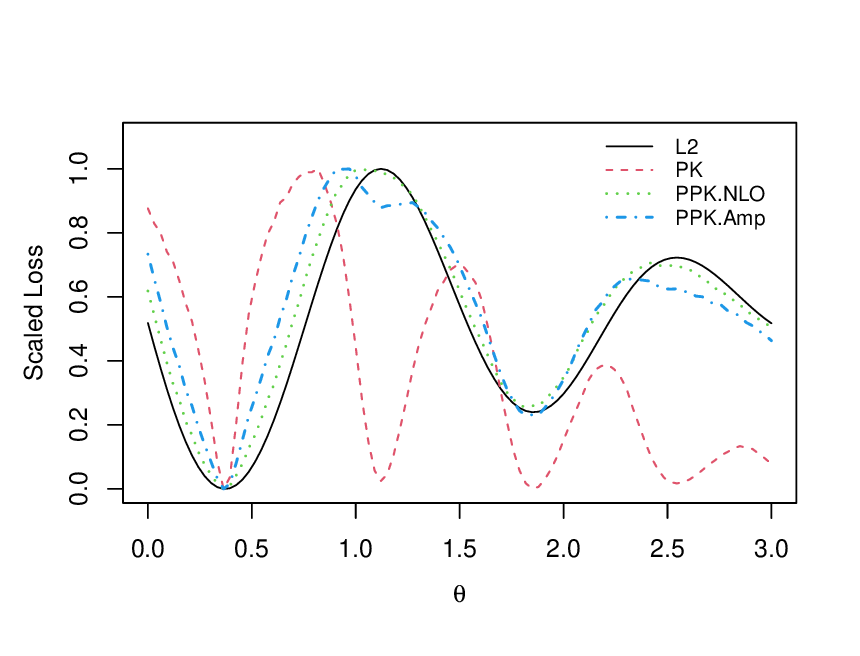}
\end{minipage}%
}%
\subfigure[$n=100$.]{
\begin{minipage}[t]{0.45\linewidth}
\centering
\includegraphics[width=2.7in]{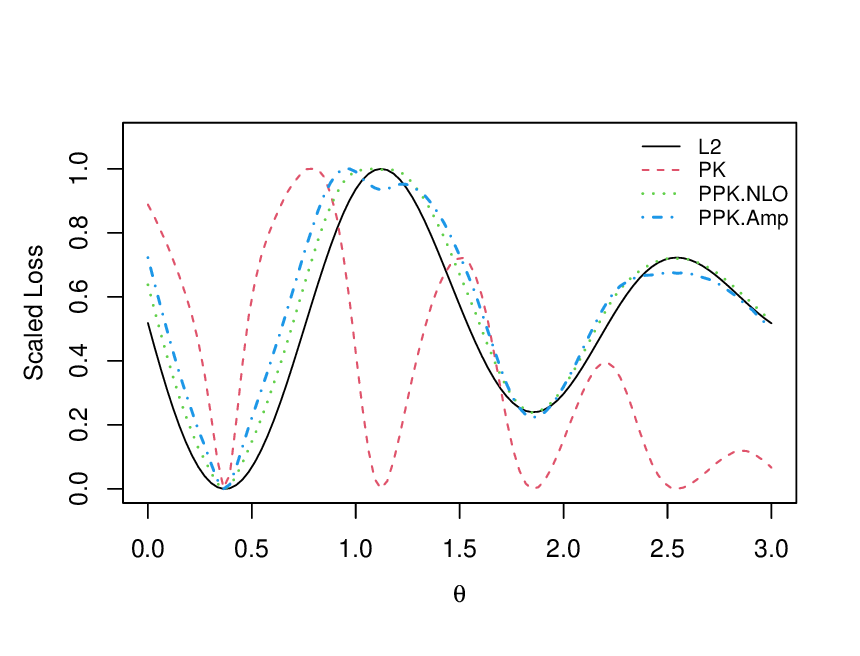}
\end{minipage}%
}%
\centering
 \caption{ (a) the true process $\zeta(x)$ (black real line) v.s. the calibrated computer model $y^s(x,\theta^*)$ (black dashed line); (b)-(d) $L_2$ loss function (black real line) v.s. PK loss function (red dashed line) v.s. PPK. NLO loss function (green dotted line) v.s.  PPK. Amp loss function  (blue dot-dash line).}
   \label{fig:ppk2}
\end{figure}

 Looking at Fig.  \ref{fig:ppk2}-(a), one may see that even though $\theta^*$ is the global minima of  $\|\zeta-y^s(\cdot,\theta)\|^2_{L_2(\Omega)}$, the discrepancy between the true process and the calibrated computer model is still large.  Figs. \ref{fig:ppk2}-(b-d) 
provide a comparison among the different loss functions when the sample sizes are $6, 15$ and $100$. PPK. NLO loss function is the PPK loss function determined by using the $NLO$ index, whereas PPK.Amp loss is that found using the $Amp$ index. Since we have a single parameter, we use the package \emph{rootSolve} \cite{soetaert2009rootsolve} in R to 
obtain the number of local optima of the loss function and the estimation of $\eta_{NLO}$. Let $\Theta_s=\{\theta_1,\ldots,\theta_{N'}\}$ where  $\theta_i=\frac{3i}{N'}$ and $N'=100$.  By approximating $Amp(L)$ using (\ref{disamp}), we obtain $\eta_{Amp}$.  Table \ref{tab:gamma} summarizes results for $\eta_{NLO}$ and $\eta_{Amp}$ at different
sample sizes.

\begin{table}[!htbp]
    \caption{Choice of $\eta$}
    \label{tab:gamma}
    \centering
       \setlength{\tabcolsep}{3pt}
    \renewcommand{\arraystretch}{1.2}
  {  \begin{tabular}{lccccccc}
        \hline
        Sample size&$\eta_{NLO}$&$\eta_{Amp}$&\\
        \hline
        $n=6$ & $0.153$ & $0.0503$  \\
         $n=15$ & $0.705$&  $0.236$\\
       $n=100$& $13.632$ & $6.477$\\
           \hline
    \end{tabular}}
 \end{table}
Figure \ref{fig:ppk2} shows that when the sample size is larger than $6$, the PPK loss function has several extrema. Figure \ref{fig:ppk2}-(b) shows that when the sample size is $6$, one may discriminate the global minimum from the local minimum near $1.855$ by an effective optimization algorithm. However, Figs. \ref{fig:ppk2}-(c-d) show that when the sample size is larger than $15$, it becomes extremely hard to find the global minimum by any optimization algorithm. It can be seen that using the PPK loss function solves this problem. The number of local optima for the PPK.NLO loss functions is $4$, and the 
values of the local minima are different. Although the PPK.Amp loss functions have more than $4$ local optimal points, the global minimizer may be evaluated effectively by some optimization method, e.g. the quasi-Newton optimization methods with multiple initial points.

 As it follows from its definition, and from the fact that $NLO(L)$ counts the 
number of local optima, the BIC-like criterion looks for values of 
$\eta$ that decrease the number of local optimal points of the PPK loss 
function. As a consequence, $\eta_{NLO}$ is larger than $\eta_{Amp}$ as shown in Table  \ref{tab:gamma}, and PPK.NLO loss functions have less local optima than the PPK.Amp one,  as shown in Figure \ref{fig:ppk2}. Moreover,  $\eta_{NLO}$ and $\eta_{Amp}$ are increasing with the sample size. The reason is that it becomes much harder to pick out the global minimum of the PK loss function for increasing $n$.

Let us now compare the performance of the PPK calibration with that of KO's calibration (KO), $L_2$ calibration ($L_2$), least square calibration (LS), scaled Gaussian process model calibration (SGP), and projected kernel calibration (PK). To this aim, we repeat the process of calibration 100 times for each method, and show the box-plots of $\hat\theta$ in Figure \ref{fig:ppk3}. Since the PK calibration is easily trapped in a local optimal solution of the PPK loss function, we narrow the search space of the PK calibration to $[0,0.5]$.  
 \begin{figure}[htbp]
\centering
\subfigure[$n=6$.]{
\begin{minipage}[t]{0.95\linewidth}
\centering
\includegraphics[width=4.9in]{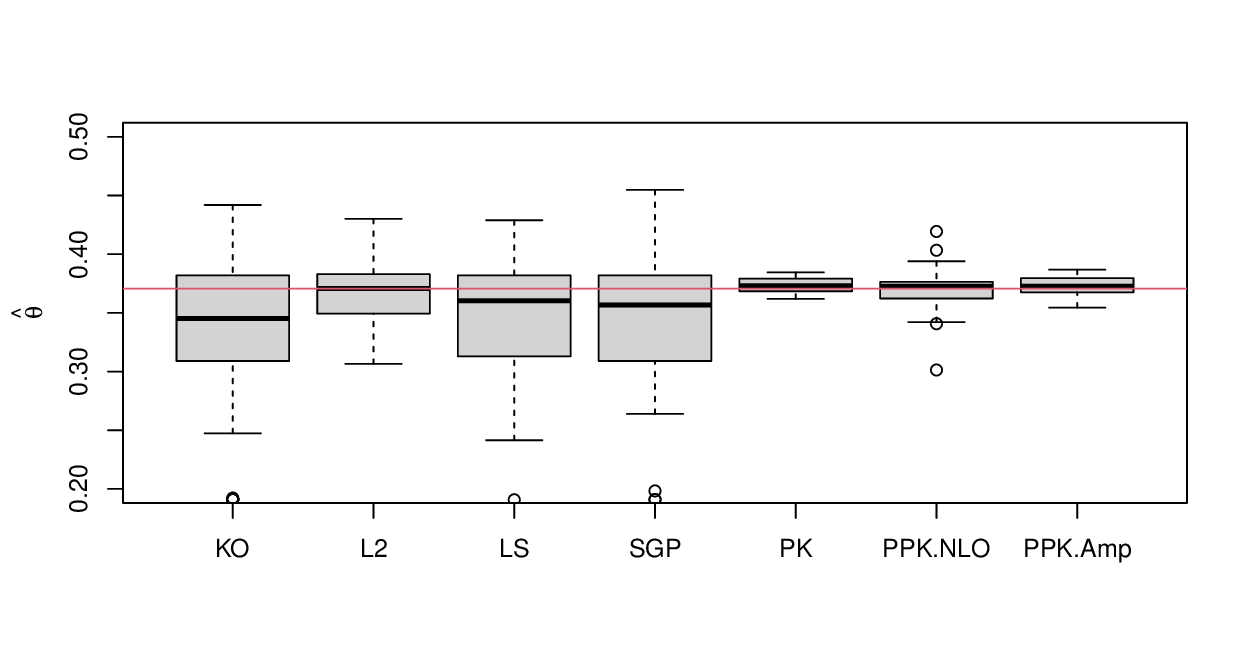}
\end{minipage}%
}%
\quad
\subfigure[$n=15$.]{
\begin{minipage}[t]{0.95\linewidth}
\centering
\includegraphics[width=4.9in]{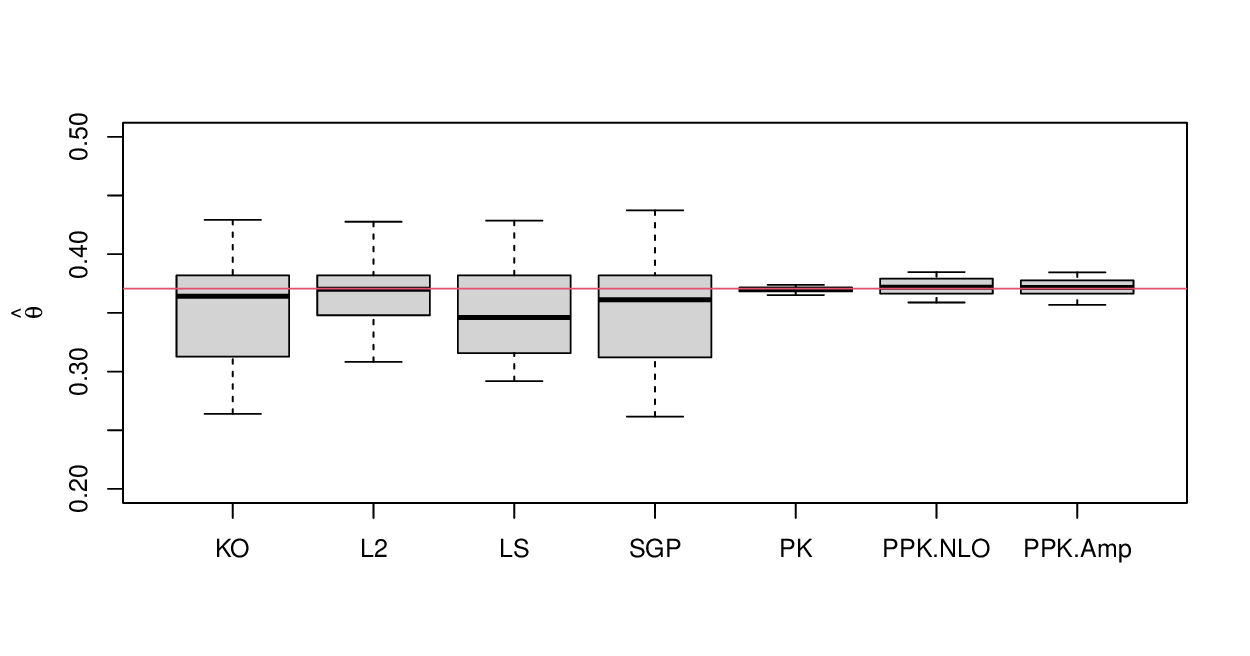}
\end{minipage}%
}%
\quad
\subfigure[$n=100$.]{
\begin{minipage}[t]{0.95\linewidth}
\centering
\includegraphics[width=4.9in]{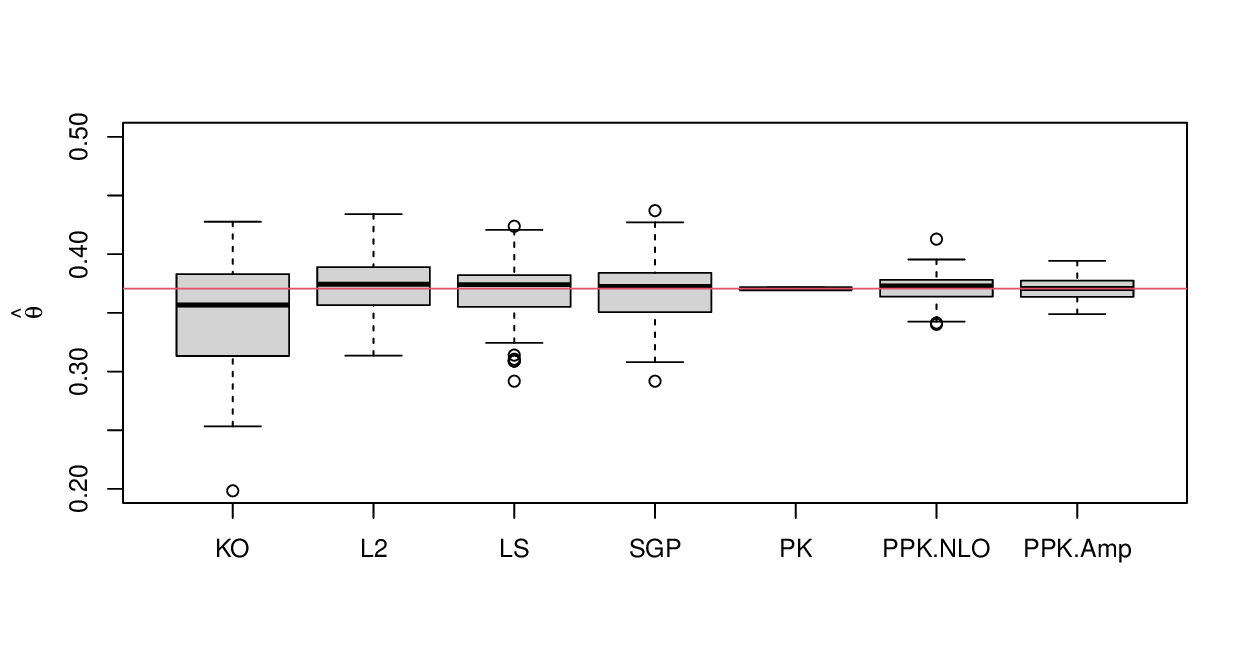}
\end{minipage}%
}%

\centering
 \caption{Estimations of the calibration parameter by different methods.}
   \label{fig:ppk3}
\end{figure}

 From Fig. \ref{fig:ppk3}, we can see  that the variance of $\hat\theta_{PK}$ is the smallest. In fact, as the PK loss function tends to have more local optima, the PK loss function near its global minimum point is ``sharp"  (the derivative near the global minimum point is large). We also see that the bias of $\hat\theta_{PK}$ is close to zero.  These results indicate that if the search region is narrowed, the PK calibration is clearly superior to the other methods. However, when there are many local optima in the search region, the PK calibration loses its advantages.

When the sample size is $6$ and $15$, the slower convergence speed leads to poor performance of the $\hat\theta_{SGP}$.
When the sample size is $100$, we can still see substantial estimation errors in $\hat\theta_{KO}$, the reason is that the discrepancy $\delta$ is large ( see Figure \ref{fig:ppk2}-(a)) and $\hat\theta_{KO}$ is inconsistent. 
 
If we denote the estimate obtained from PPK.NLO (PPK.Amp) as $\hat\theta_{PPK.NLO}$ ( $\hat\theta_{PPK.Amp}$), we have that the bias of $\hat\theta_{PPK}$ is close to zero. 
The variance of $\hat\theta_{PPK}$ is smaller than the variance of $\hat\theta$ given by other methods, except that from the PK calibration method. In addition, because the tuning parameter $\eta_{Amp}<\eta_{NLO}$, the $\hat\theta_{PPK.Amp}$'s are  closer to  $\hat\theta_{PK}$, and the variance of  $\hat\theta_{PPK.Amp}$ is smaller than the variance of  $\hat\theta_{PPK.NLO}$. It implies that our proposed method outperforms the other calibration methods.

  \subsection{Low-accuracy version of the PARK function \cite{xiong2013sequential}}
  \label{ex2}
Assume that $\zeta(\cdot)$ is the PARK function \cite{park1991tuning},
 $$ \zeta{(\bm x)}=\frac{x_1 } {2} \left[\sqrt{1 + (x_2+x_3^2)\frac{x_4}{x_1^2}} - 1\right]+(x_1 + 3x_4) \exp[1 + \sin(x_3)], \bm x\in[0,1]^4.$$
In Ref. \cite{xiong2013sequential} a lower accuracy version of the PARK function is used
for the purpose of multi-fidelity simulation. Assuming that some constants of this  lower fidelity model are to be determined, we use following computer model to examine the performance of the proposed method:
 $$y^s(\bm x,\bm \theta)=(\theta_1+\frac{\sin(x_1)}{10}) \zeta(\bm x) +\theta_2( -2x_1 + x_2^2 +x_3^2)+0.5,$$
where $\theta_1$ and $\theta_2$ are two calibration parameters, with $\bm\theta \in [-5,5]^2$.   

Let $\mathbf X = (\bm x_1,\ldots, \bm x_n)^T$ be the physical design, which is randomly generated by maximin Latin hypercube design method \cite{santner2013design}.  Suppose the observation error $\epsilon_i$'s are mutually independent and distributed as $N(0,0.1^2)$. 
We use a Mat\'ern kernel function (\ref{matern f}) with $\nu=7/2$ as the kernel function $K$. To determine the hyper-parameter $\rho$ in (\ref{matern f}), for fixed $\bm\theta_0$, we build a Gaussian-process model to approximate $y(\bm x_i)-y^s(\bm x_i,\bm\theta_0)$ and estimate $\rho$ by using maximum likelihood. Because the least square estimator $\hat{\bm\theta}_{LS}$ proposed in \cite{wong2014frequentist} is consistent, we set $\bm\theta_0=\hat{\bm\theta}_{LS}$. 

 \begin{figure}[htbp]
  \centering
  \includegraphics[scale=0.65]{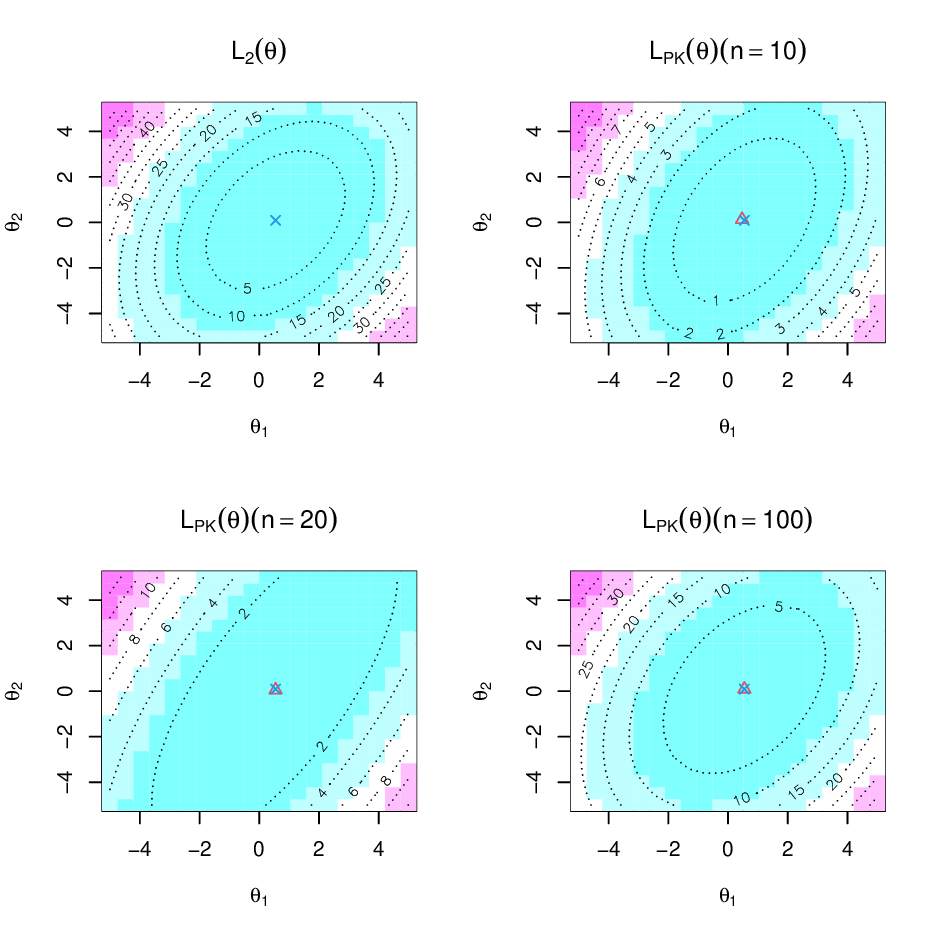}
  \caption{$L_2$ loss function v.s. PK loss function with $n=\{10,20,100\}$ in Example \ref{ex2}.The cross in each subfigure is the location of $\bm\theta^*$, and triangle is the location of $\bm\theta_{PK}$.}
  \label{fig:l2loss.park}
\end{figure}

Contour maps of the $L_2$ and the PK loss functions with $n=\{10, 20,100\}$ are shown in Figure \ref{fig:l2loss.park}. From the top left subfigure, we can see that, $\bm\theta^*$ is the only local optimal point of the $L_2$ loss function. The top right subfigure, and the two lower subfigures, show that the PK loss function has only one local minimum, regardless of the sample size. This indicates that when the $L_2$ loss function has only one local optimal point, the PK loss function is not affected by the multiple local minima problem. Since the $L_2$ and the PPK loss functions are convex, we apply the NEWUOA algorithm \cite{powell2006newuoa} to find $\bm \theta^*$ and $\hat{\bm\theta}_{PK}$.  In Fig. \ref{fig:l2loss.park}, $\bm\theta^*=(0.546,0.0926)$ and the $\hat{\bm\theta}_{PK}$ are denoted by a blue cross and by red triangles, respectively. By
comparing the locations of the red triangles with that of the blue cross, we have that, $\hat{\bm\theta}_{PK}$ is very close to $\bm\theta^*$ especially then $n$ is large.

 To compare the performance of the proposed method with some existing calibration methods, we repeat the simulation procedure $100$ times to assess the average performance of different methods. 
\begin{figure}[htbp]
  \centering
  \includegraphics[scale=0.65]{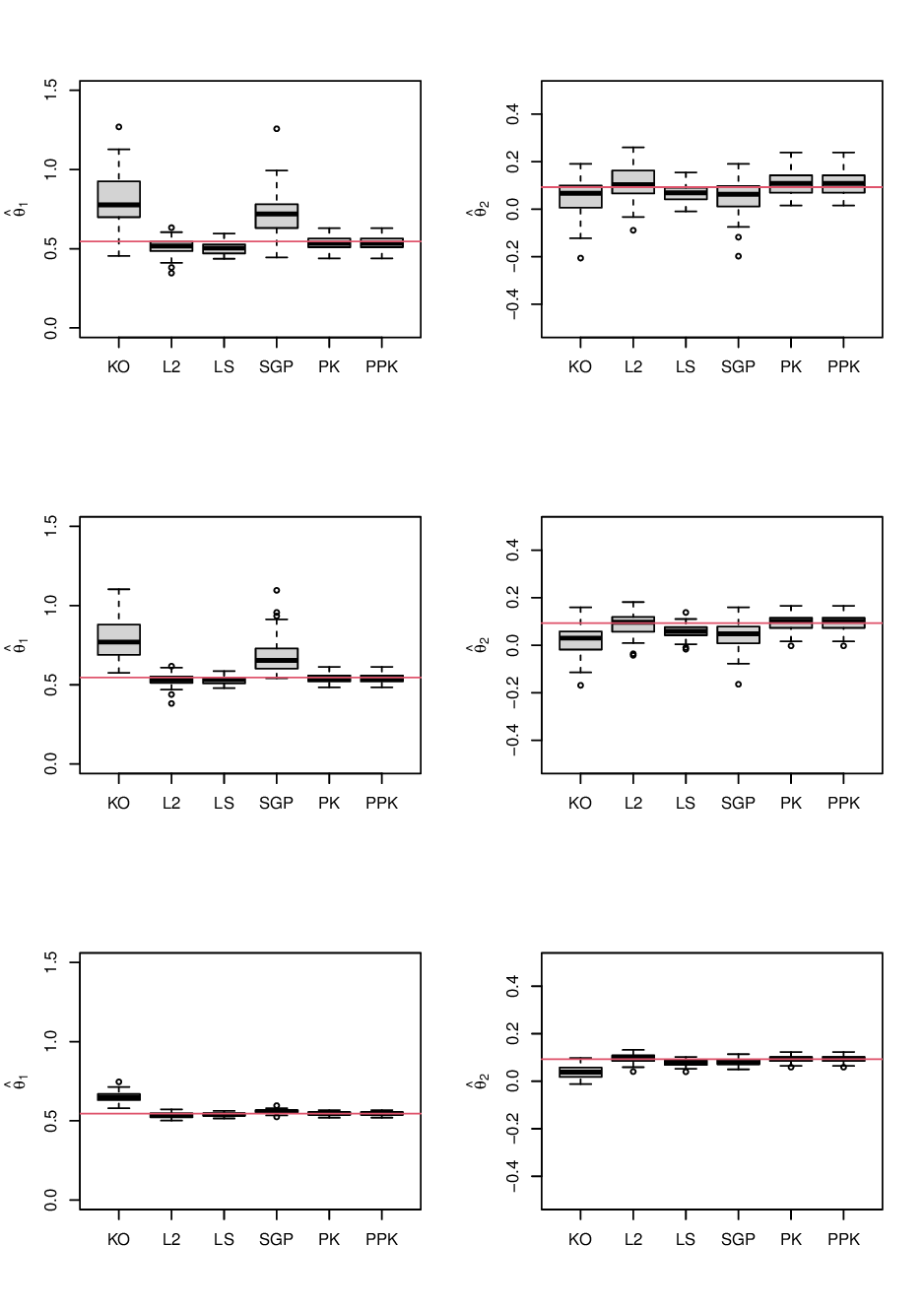}
  \caption{Comparisons of different calibration methods with the sample size $n=10$ (the first row), $n=20$ (the second row) and $n=100$ (the third row).}
  \label{fig:com.park}
\end{figure}
 
 Figure \ref{fig:com.park} illustrates the estimation results of the different calibration methods. Since there is only one local optimal point in the PK loss function, the choice of $\eta$ is zero, therefore $\hat{\bm\theta}_{PPK}=\hat{\bm\theta}_{PK}$. We can see that when the sample size is $10$ and $20$, the bias of $\hat{\bm\theta}_{PK}$ is the smallest, whereas the variance of $\hat{\bm\theta}_{PK}$ is slightly larger than $\hat{\bm\theta}_{LS}$.  When the sample size is $100$, all the methods perform well except the KO's calibration. 

 \subsection{Spot welding example}
Let us now consider the spot welding example studied in \cite{bayarri2012framework} and \cite{xie2020bayesian}. Analogously to \cite{xie2020bayesian}, we consider  two control variables: the load and the current.  Besides the control variables in the physical
experiment, the computer model (a Finite Element model) also involves a calibration parameter (denoted as $u$
in \cite{bayarri2012framework}). Details of the inputs and outputs of the computer experiments are listed in Table  \ref{tab:spotweld}.
  \begin{table}[!htbp]
   \caption{Inputs and output of the computer experiments}
    \label{tab:spotweld}
    \centering
    \footnotesize
    \setlength{\tabcolsep}{4pt}
    \renewcommand{\arraystretch}{1.2}
    \begin{tabular}{lcccc}
        \hline
       Inputs &\\
        \hline
      $C$ ( current )& $[23,30]$ &control  variable \\
      $L$ (load)  & $[3.8,5.5]$ &control variable \\
        $\theta$ (contact resistance) & $[0.8,8]$ & calibration parameter \\
        \hline
         Output& Size of the nugget after 8-cycles\\
          \hline
          \end{tabular}
\end{table}
The physical data are listed in Table 4 of \cite{bayarri2012framework}. There are
21 available runs for the computer code, as presented in Table 3 of \cite{bayarri2012framework}.
With the help of the RobustGaSP package \cite{gu2018robustgasp} in R, a Gaussian process model is built to approximate the computer outputs.  In the process of calibration, the Finite Element model is replaced by the predictive mean of the RobustGaSP emulator. Since there is only one local optimal point in the PK loss function, also here we have $\eta=0$. 
 \begin{figure}[htbp]
  \centering
  \includegraphics[scale=0.65]{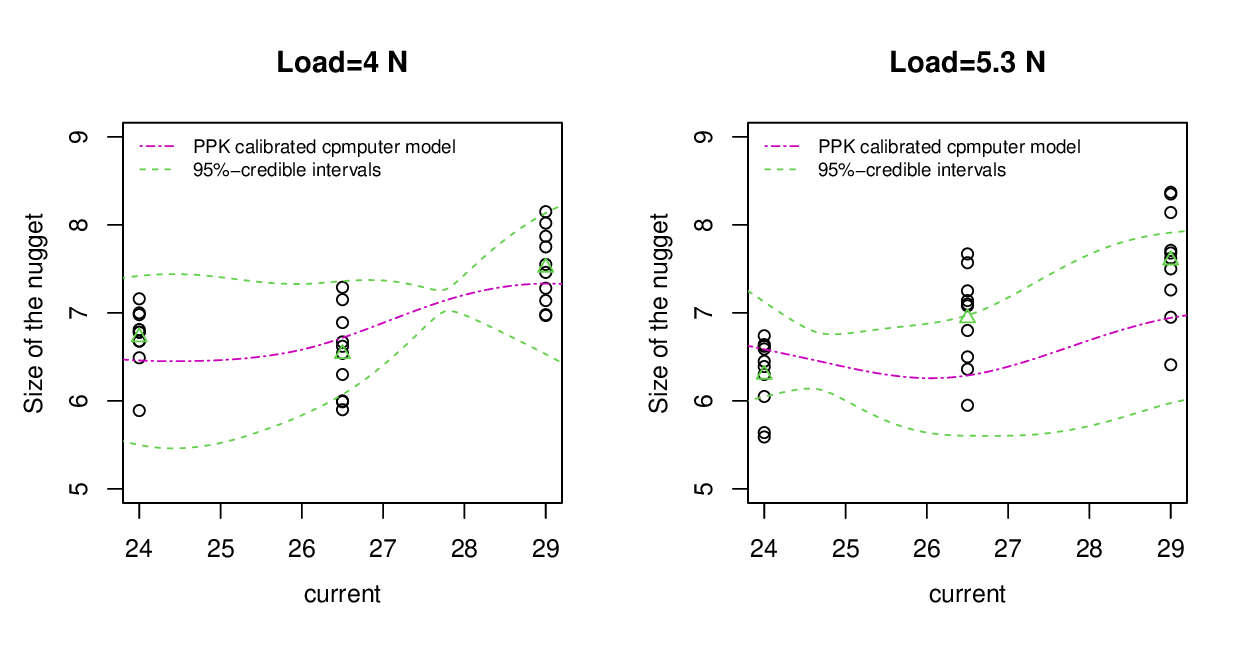}
  \caption{Physical observations (black circles); mean of physical observations for a fixed current (green triangles);  mean of the calibrated computer model by the proposed method (red real line) and the 95\% interval of  the calibrated computer model (green dashed line).}
  \label{fig:sw}
\end{figure}
 
The computer models calibrated by the proposed method, together with their corresponding point-wise 95\%-credible
intervals, are depicted in Figure \ref{fig:sw}. We can see that the proposed method provides a well calibrated computer model.

\section{Discussion}
\label{sec6}
In this work, we have proven that the projected kernel calibration may be easily trapped 
in local minima of the $L_2$ loss between the true process and the computer model (or even in local maxima). A frequentist calibration method has been proposed to overcome this problem. The estimators of the calibration parameters given by the proposed method are consistent and semi-parametrically efficient. Numerical examples have been studied to compare the proposed methods with some existing calibration methods, and results show that our  method outperforms the others.

{The proposed method suggests  a Monte Carlo method   (\ref{approx}) to approximate the $L_2$ inner products.  
However, there is no guarantee that the numerical estimator based on the Monte Carlo approximation is close enough to the theoretical estimator  (\ref{ppkestimator}). This will require further work.

In this work, we assume that all possible functions of interest belong to the reproducing kernel Hilbert space $\mathcal{N}_K(\Omega)$ generated by the kernel function $K$. We also assume that the space $\mathcal{N}_K(\Omega)$ can be embedded into the Sobolev space $H^m(\Omega)$ with $m>d/2$. In other words, we are mainly concerned with the reproducing kernel Hilbert spaces  generated by the smooth kernel functions. Rough kernel functions  such as those generated by the rough fractional maximal and integral operators \cite{gurbuz2017some, gurbuz2020behaviors} and the Schr{\"o}dinger operators \cite{gurbuz2018generalized, gurbuz2020generalized, gurbuz2021note}, etc. are not covered in this paper. The results for calibration with the rough kernel functions  need further investigation. }

\appendix

 \section{Technical  proofs in Section \ref{comparison}}
\label{App:proof-2}
In this section, we prove the propositions and theorems in Section \ref{comparison}.

 \subsection{Proof of Proposition \ref{th:equatheta}}
\begin{proof}
	Following from the generalized representer’s theorem \cite{scholkopf2001generalized}, $\hat\delta^{\theta}_{PK}$ can be represented by 
  \begin{eqnarray}
  \label{krr-repres}
\hat\delta^{\theta}_{PK}(\bm x)=K^T_{\theta}(\bm x,\bm X)(\mathbf K_{\theta}+n\lambda \mathbf{I}_n)^{-1}(\bm Y-\bm Y^s_{\theta}),
\end{eqnarray}
{where $K_{\theta}(\bm x,\bm X)=(K_{\theta}(\bm x,\bm x_1),\ldots,K_{\theta}(\bm x,\bm x_n))^T$.}
{Therefore the norm of $\hat{\delta}_{PK}^\theta$ in $\mathcal{N}_{K_{
	\theta}}(\Omega)$ is given by}
	\begin{eqnarray}\label{native}
	\|\hat{\delta}_{PK}^\theta\|^2_{\mathcal{N}_{K_{
	\theta}}(\Omega)}=(\bm Y-\bm Y_{\theta}^s)^T (\mathbf K_{\theta}+n\lambda \mathbf{I}_n)^{-1} \mathbf K_{\theta}(\mathbf K_{\theta}+n\lambda \mathbf{I}_n)^{-1} (\bm Y-\bm Y_{\theta}^s).
	\end{eqnarray}
	
	Moreover, because the vector $\left(\hat{\delta}_{PK}^\theta(\bm x_1),\ldots,\hat{\delta}_{PK}^\theta(\bm x_n)\right)^T$ can be expressed by
	\begin{eqnarray*}
		\mathbf K_{\theta}(\mathbf K_{\theta}+n\lambda \mathbf{I}_n)^{-1} (\bm Y-\bm Y_{\theta}^s).
	\end{eqnarray*}
	It yields
	\begin{eqnarray}\label{n}
	\frac{1}{n}\sum_{i=1}^n (\delta_i^\theta-\hat{\delta}_{PK}^\theta(\bm x_i))^2=n\lambda^2(\bm Y-\bm Y_{\theta}^s)^T (\mathbf K_{\theta}+n\lambda I_n)^{-2} (\bm Y-\bm Y_{\theta}^s).
	\end{eqnarray}
	
	Combining (\ref{native}) with (\ref{n}), we obtain that
	\begin{eqnarray*}
		\frac{1}{n}\sum_{i=1}^n (\delta_i^\theta-\hat{\delta}_{PK}^\theta(\bm x_i))^2+\lambda\|\hat{\delta}_{PK}^\theta\|^2_{\mathcal{N}_{K_{\theta}}(\Omega)}=\lambda (\bm Y-\bm Y_{\theta}^s)^T(\mathbf K_{\theta}+n\lambda I_n)^{-1} (\bm Y-\bm Y_{\theta}^s),
	\end{eqnarray*}
	which implies the desired result.
\end{proof}

\subsection{Proof of Proposition \ref{th:lossforderiv}}
\begin{proof}
Because  $\mathcal{G}_{\theta}\subset \mathcal{N}_K(\Omega) $ is finite dimensional,  based on the generalized representer’s theorem \cite{scholkopf2001generalized}, $\hat\delta^{\theta}$ can be represented as
\begin{eqnarray}
\label{gen-rep}
\hat\delta^{\theta}(\bm x)=\sum_{i=1}^n \alpha^\theta_i K_{\theta}(\bm x,\bm x_i)+\sum_{j=1}^q \beta^\theta_j \frac{\partial y^s(\bm x,\bm\theta)}{\partial \theta_j},
\end{eqnarray}
with $\bm \alpha^{\theta}=(\alpha_1^\theta,\ldots,\alpha_n^\theta)^T$  and  $\bm \beta^{\theta}=(\beta_1^\theta,\ldots,\beta_q^\theta)^T$ defined as
\begin{eqnarray*}
		\bm \alpha^{\theta}=(\mathbf K_{\theta}+n\lambda \mathbf{I}_n)^{-1}(\bm Y-\bm Y^s_{\theta}),
	\end{eqnarray*}
and 	
\begin{eqnarray*}
		\bm \beta^{\theta}=\left(\mathbf G_{\theta}\mathbf G^T_{\theta}\right)^{-1}\mathbf G_{\theta}(\bm Y-\bm Y^s_{\theta}),
	\end{eqnarray*}
	respectively. Here $\mathbf G_{\theta}$ is a $q\times n$ matrix, with $[\mathbf G_{\theta}]_{j,i}=\frac{\partial y^s(\bm x_i,\bm\theta)}{\partial \theta_j}$, $j=1,\ldots,q; i=1\ldots,n$ and $q\leq n$.

From (\ref{gen-rep}), we can easily have that
$$\mathcal{P}^{\perp}_{\mathcal{G}_{\theta}}\hat\delta^{\theta}=\sum_{i=1}^n \alpha^\theta_i K_{\theta}(\bm x,\bm x_i),$$
which implies the desired results.	
\end{proof}

{  \subsection{Proof of Proposition \ref{condelta}}
 \begin{proof}
{Because $\hat{\delta}^{\theta}$   is a minimizer of 
\begin{eqnarray}
l_{\theta}(\delta_0)=\frac{1}{n}\sum_{i=1}^n\left(\delta_i^{\theta}-\mathcal{P}_{\mathcal{G}_{\theta}}\delta_0(\bm x_i)\right)^2+\frac{1}{n}\sum_{i=1}^n\left(\delta_i^{\theta}-\mathcal{P}^{\perp}_{\mathcal{G}_{\theta}}\delta_0(\bm x_i)\right)^2+\lambda\|\mathcal{P}^{\perp}_{\mathcal{G}_{\theta}}\delta_0\|^2_{\mathcal{N}_{K_{\theta}}(\Omega)},
\end{eqnarray}
where $\delta_0 \in \mathcal{N}_K(\Omega)$,
we can deduce the following basic inequality}
 \begin{eqnarray*}
 \begin{aligned}
&\frac{1}{n}\sum_{i=1}^n\left(\delta_i^{\theta}-\mathcal{P}_{\mathcal{G}_{\theta}}\hat{\delta}^{\theta}(\bm x_i)\right)^2+\frac{1}{n}\sum_{i=1}^n\left(\delta^{\theta}_i-\hat{\delta}^{\theta}_{PK}(\boldsymbol x_i)\right)^2+\lambda\|\hat{\delta}^{\theta}_{PK}\|^2_{\mathcal{N}_{K_{\theta}}(\Omega)},\\
&\leq \frac{1}{n}\sum_{i=1}^n\left(\delta_i^{\theta}-\mathcal{P}_{\mathcal{G}_{\theta}}{\delta}^{\theta}(\bm x_i)\right)^2+ \frac{1}{n}\sum_{i=1}^n\left(\delta^{\theta}_i-\mathcal{P}^{\perp}_{\mathcal{G}_{\theta}}{\delta}^{\theta}(\boldsymbol x_i)\right)^2+\lambda\|\mathcal{P}^{\perp}_{\mathcal{G}_{\theta}}{\delta}^{\theta}\|^2_{\mathcal{N}_{K_{\theta}}(\Omega)},\\
\end{aligned}
\end{eqnarray*}
which holds for all $\bm\theta\in\Theta$.
With some simple calculations, the basic inequality can be expressed as
 \begin{eqnarray}
  \label{basicineq}
 \begin{aligned}
&2\left<\mathcal{P}^{\perp}_{\mathcal{G}_{\theta}}{\delta}^{\theta}-\hat{\delta}^{\theta}_{PK}, \mathcal{P}_{\mathcal{G}_{\theta}}{\delta}^{\theta}\right>_n+
2\left<\mathcal{P}_{\mathcal{G}_{\theta}}({\delta}^{\theta}-\hat{\delta}^{\theta}), \mathcal{P}^{\perp}_{\mathcal{G}_{\theta}}{\hat \delta}^{\theta}\right>_n+\\&\left\|{\delta}^{\theta}-\hat{\delta}^{\theta}\right\|_n^2+\lambda\|\hat{\delta}^{\theta}_{PK}\|^2_{\mathcal{N}_{K_{\theta}}(\Omega)}\\
&\leq \lambda\|\mathcal{P}^{\perp}_{\mathcal{G}_{\theta}}{\delta}^{\theta}\|^2_{\mathcal{N}_{K_{\theta}}(\Omega)}+2\left|\left<{\delta}^{\theta}-\hat{\delta}^{\theta}, \epsilon\right>_n\right|.\\
\end{aligned}
\end{eqnarray}

{Next we bound the first two terms on the left side of  (\ref{basicineq}) and the two terms   on the right side of  (\ref{basicineq}),  respectively. }

\begin{itemize}

\item For the first and the second terms on the left side of the basic inequality  (\ref{basicineq}), because $x_i$’s follow the uniform distribution over $\Omega$, there is an asymptotic equivalence relation between the $L_2$ and the empirical norm \cite{tuo2019adjustments}:
  \begin{eqnarray} 
  \label{semi2}
\lim_{n\rightarrow \infty}\sup P\left\{\sup_{\|g\|_{\mathcal{N}_{K}(\Omega)}=O_p(1), \|g\|_{L_2(\Omega)}>\tau n^{-\frac{m}{2m+d}}/\eta}\left| \frac{\|g\|_n}{\|g\|_{L_2(\Omega)}}-1\right|\geq \eta\right\}=0.
 \end{eqnarray}
As a result, we have that
 \begin{eqnarray}
 \label{firstterm} 
\left<\mathcal{P}^{\perp}_{\mathcal{G}_{\theta}}{\delta}^{\theta}-\hat{\delta}^{\theta}_{PK}, \mathcal{P}_{\mathcal{G}_{\theta}}{\delta}^{\theta}\right>_n+\left<\mathcal{P}_{\mathcal{G}_{\theta}}({\delta}^{\theta}-\hat{\delta}^{\theta}), \mathcal{P}^{\perp}_{\mathcal{G}_{\theta}}{\hat \delta}^{\theta}\right>_n=o_p(n^{-1/2}).
\end{eqnarray}

\item  For the first term on the right side of the basic inequality  (\ref{basicineq}), following from the Theorem 3.3 in \cite{tuo2019adjustments} and together with the condition A3, we have that, there is a constant $a_1>0$ such that 
{\begin{eqnarray}
\label{sec-right}
\|\mathcal{P}^{\perp}_{\mathcal{G}_{\theta}}{\delta}^{\theta}\|^2_{\mathcal{N}_{K_{\theta}}(\Omega)}\leq b^2_1\sup_{\bm\theta\in\Theta} \|{\delta}^{\theta}\|^2_{\mathcal{N}_{K}(\Omega)}\leq a_1,
\end{eqnarray}
where $$b_1=1+\sup_{\bm\theta\in \Theta}\sup_{g\in \mathcal{G}_{\theta}, \|g\|_{L_2(\Omega)}=1 }\|g\|_{\mathcal{N}_{K}(\Omega)}\|<K,g>_{L_2(\Omega)}\|_{\mathcal{N}_{K}(\Omega)}.$$}

\item  For the second term on the right side of the basic inequality  (\ref{basicineq}),
following from the Theorem 5.11 in \cite{geer2000empirical}, we obtain the modulus of continuity of the empirical process $v(g')=<\epsilon,g'-g>_n$ as
\begin{eqnarray*}
 \begin{aligned}
 \sup_{g\in \mathcal{N}_{K}(\Omega)}\frac{|<\epsilon,g'-g>_n|}{\|g-g'\|_n^{1-\frac{d}{2m}}\|g'\|^{d/2m}_{\mathcal{N}_{K}(\Omega)}}=O_p(n^{-1/2}).
\end{aligned}
\end{eqnarray*}
{That is, there is a constant $a_2>0$ such that
\begin{eqnarray}
\label{ep}
 \begin{aligned}
&|<{\delta}^{\theta}-\hat{\delta}^{\theta}, \epsilon>_n|\\
&\leq \frac{a_2}{2} n^{-1/2}\left\|{\delta}^{\theta}-\hat{\delta}^{\theta}\right\|_n^{1-\frac{d}{2m}}\left\|\delta^{\theta}\right\|^{d/2m}_{\mathcal{N}_{K}(\Omega)}.
\end{aligned}
\end{eqnarray}
By plugging  (\ref{sec-right}) into (\ref{ep}), we have that
\begin{eqnarray}
\label{ep-2}
 \begin{aligned}
&|<{\delta}^{\theta}-\hat{\delta}^{\theta}, \epsilon>_n|\\
&\leq \frac{a_2}{2} n^{-1/2}\left\|{\delta}^{\theta}-\hat{\delta}^{\theta}\right\|_n^{1-\frac{d}{2m}}\left(\frac{\sqrt{a_1}}{b_1}\right)^{d/2m}.
\end{aligned}
\end{eqnarray}}
\end{itemize}

{Let $a_3$ be the positive constant ${a_2}\left(\frac{\sqrt{a_1}}{b_1}\right)^{d/2m}$.
By combining  (\ref{firstterm}), (\ref{sec-right}) and (\ref{ep-2}), we have that  the basic inequality  (\ref{basicineq}) can be represented by
 \begin{eqnarray}
 \label{PQ}
 \begin{aligned}
\left\|{\delta}^{\theta}-\hat{\delta}^{\theta}\right\|_n^2+\lambda\|\hat{\delta}^{\theta}_{PK}\|^2_{\mathcal{N}_{K_{\theta}}(\Omega)}\leq a_1 \lambda+{a_3} n^{-1/2}\left\|{\delta}^{\theta}-\hat{\delta}^{\theta}\right\|_n^{1-\frac{d}{2m}}.
\end{aligned}
\end{eqnarray}
 Next we consider  two different cases separately. 
 
 Case I. Suppose $a_1 \lambda \geq {a_3} n^{-1/2}\left\|{\delta}^{\theta}-\hat{\delta}^{\theta}\right\|_n^{1-\frac{d}{2m}}$. 
 Then we obtain from    (\ref{PQ}) that
 \begin{equation*}
	\left\|{\delta}^{\theta}-\hat{\delta}^{\theta}\right\|_n^2+\lambda\|\hat{\delta}^{\theta}_{PK}\|^2_{\mathcal{N}_{K_{\theta}}(\Omega)}\leq 2 a_1 \lambda.
	\end{equation*}
		It implies that  the following inequalities
	\begin{eqnarray}
	\label{le2part1}
		\begin{aligned}
			\left\|{\delta}^{\theta}-\hat{\delta}^{\theta}\right\|_n^2&\leq 2a_1 \lambda
		\end{aligned}
	\end{eqnarray}
{and	}
		\begin{eqnarray}
	\label{le2part1-2}
		\begin{aligned}
			\|\hat{\delta}^{\theta}_{PK}\|^2_{\mathcal{N}_{K_{\theta}}(\Omega)}&\leq 2 a_1
		\end{aligned}
	\end{eqnarray}
 hold simultaneously.
	
	
By combining (\ref{le2part1}) and (\ref{le2part1-2}), we have that if $\lambda\sim n^{-\frac{2m}{2m+d}}$, there are constants $a_4>0$ and $a_5>0$ such that the following inequalities hold  simultaneously
	\begin{eqnarray}\label{le3}
		\begin{aligned}
	&\left\|{\delta}^{\theta}-\hat{\delta}^{\theta}\right\|_n\leq   a_4 n^{-\frac{m}{2m+d}},\\
	&\|\hat{\delta}^{\theta}_{PK}\|_{\mathcal{N}_{K_{\theta}}(\Omega)}\leq  a_5.
	\end{aligned}
	\end{eqnarray}

 Case II. Suppose $a_1 \lambda < {a_3} n^{-1/2}\left\|{\delta}^{\theta}-\hat{\delta}^{\theta}\right\|_n^{1-\frac{d}{2m}}$.  
 Then we obtain from    (\ref{PQ}) that
 \begin{equation*}
	\left\|{\delta}^{\theta}-\hat{\delta}^{\theta}\right\|_n^2+\lambda\|\hat{\delta}^{\theta}_{PK}\|^2_{\mathcal{N}_{K_{\theta}}(\Omega)}\leq 2{a_3} n^{-1/2}\left\|{\delta}^{\theta}-\hat{\delta}^{\theta}\right\|_n^{1-\frac{d}{2m}}.
	\end{equation*}
 	It implies that  the following inequalities
	\begin{eqnarray}\label{PQineq2}
	\begin{aligned}
	\left\|{\delta}^{\theta}-\hat{\delta}^{\theta}\right\|_n^2&\leq 2 {a_3} n^{-1/2}\left\|{\delta}^{\theta}-\hat{\delta}^{\theta}\right\|_n^{1-\frac{d}{2m}}
		\end{aligned}
	\end{eqnarray}
	{ and}
		\begin{eqnarray}\label{PQineq2-2}
	\begin{aligned}
	\lambda\|\hat{\delta}^{\theta}_{PK}\|^2_{\mathcal{N}_{K_{\theta}}(\Omega)}&\leq 2 {a_3} n^{-1/2}\left\|{\delta}^{\theta}-\hat{\delta}^{\theta}\right\|_n^{1-\frac{d}{2m}}
	\end{aligned}
	\end{eqnarray}
hold  simultaneously.
	
	From (\ref{PQineq2}) and (\ref{PQineq2-2}), we have that  if $\lambda\sim n^{-\frac{2m}{2m+d}}$, there are constants $a_6>0$ and $a_7>0$ such that the following inequalities hold  simultaneously
	\begin{eqnarray}\label{le4}
		\begin{aligned}
	&\left\|{\delta}^{\theta}-\hat{\delta}^{\theta}\right\|_n\leq   a_6  n^{-\frac{m}{2m+d}},\\
	&\|\hat{\delta}^{\theta}_{PK}\|_{\mathcal{N}_{K_{\theta}}(\Omega)}\leq  a_7.
	\end{aligned}
	\end{eqnarray}
The desired results then follow by combining
 (\ref{le3}) and (\ref{le4}).}

  \end{proof}

\subsection{Proof of Theorem  \ref{th:stationary points}}
\label{1st der}
\begin{proof} 
The first derivative of the PK loss function on $\theta_j, j=1,\ldots,q$ can be evaluated by
   \begin{eqnarray} 
    \begin{aligned}
& \frac{\partial}{\partial \theta_j} \left\{\frac{1}{n}\sum_{i=1}^n\left(\delta_i^{\theta}-\hat\delta_{PK}^{\theta}(\bm x_i)\right)^2+\lambda\|\hat\delta_{PK}^{\theta}\|^2_{\mathcal{N}_{K_{\theta}}(\Omega)}\right\}\\
&=\frac{2}{n}\sum_{i=1}^n\left(\delta_i^{\theta}-\hat\delta_{PK}^{\theta}(\bm x_i)\right)\frac{\partial\left(\delta_i^{\theta}-\hat\delta_{PK}^{\theta}(\bm x_i)\right)}{\partial\theta_j}+\lambda\frac{\partial \|\hat\delta_{PK}^{\theta}\|^2_{\mathcal{N}_{K_{\theta}}(\Omega)}}{\partial\theta_j}\\
    \end{aligned}
    \label{dpkc}
    \end{eqnarray}
   
    Next we consider the partial derivatives of  $\delta_i^{\theta}$, $ \hat \delta_{PK}^{\theta}$ and $\|\hat\delta_{PK}^{\theta}\|^2_{\mathcal{N}_{K_{\theta}}(\Omega)}$ on $\theta_j$ separately. 
   { 
    \begin{itemize}
\item  It can be easily seen that   
     \begin{eqnarray} 
  \frac{\partial \delta_i^{\theta}}{\partial \theta_j}=
 -\frac{\partial y^s(\bm x_i,\bm\theta)}{\partial \theta_j}.
 \label{ddelta}
  \end{eqnarray}
   
 \item From Proposition \ref{th:lossforderiv}, we have $\hat \delta_{PK}^{\theta}=\hat\delta^{\theta}-\mathcal {P}_{\mathcal{G}_{\theta}}\hat\delta^{\theta}$. 
Because
 $\mathcal {P}_{\mathcal{G}_{\theta}}\hat\delta^{\theta}=\bm b^{T}_{\theta}\mathbf E^{-1}_{\theta}\bm g_{\theta}$, where $\bm g^T_{\theta}(\cdot)=\frac{\partial{y^s(\cdot,\bm\theta)}}{\partial{\bm\theta}}=\left(\frac{\partial{y^s(\cdot,\bm\theta)}}{\partial{\theta_1}},\ldots, \frac{\partial{y^s(\cdot,\bm\theta)}}{\partial{\theta_q}}\right)$ and $\bm b^{T}_{\theta}=\int_{\Omega}\hat\delta^{\theta}(\bm x)\frac{\partial{y^s(\bm x,\bm\theta)}}{\partial{\bm\theta}}d\bm x$, by some elementary calculations, we have that
     
    
%

 \begin{eqnarray} 
  \begin{aligned}
  \frac{\partial \hat \delta_{PK}^{\theta}}{\partial \theta_j}&
 &= -\frac{\partial y^s(\bm x_i,\bm\theta)}{\partial \theta_j}-\frac{\partial {\bm b^T_{\theta}}}{\partial \theta_j}\mathbf E^{-1}_{\theta}\bm g_{\theta}-\bm b^T_{\theta}\frac{\partial {\mathbf E^{-1}_{\theta}}}{\partial \theta_j} \bm g_{\theta}-\bm b^T_{\theta}\mathbf E^{-1}_{\theta}\frac{\partial \bm g_{\theta}}{\partial \theta_j}.\\
\end{aligned}
\label{ddeltahat}
\end{eqnarray}

\item It following from Lemma 6.6 of \cite{tuo2019adjustments} that 
 \begin{eqnarray} 
\frac{ \partial\|\hat\delta_{PK}^{\theta}\|^2_{\mathcal{N}_{K_{\theta}}(\Omega)}}{\partial \theta_j}=O_p(1).
\label{drphsn}
 \end{eqnarray}
\end{itemize} }
 
 Plugging (\ref{ddelta}) -(\ref{drphsn}) into  (\ref{dpkc}), and also since $\lambda\sim n^{-\frac{2m}{2m+d}}$ holds, we have that the partial derivative of the PK loss function becomes

 \begin{eqnarray} 
  \begin{aligned}
&\frac{2}{n}\sum_{i=1}^n\left(\delta_i^{\theta}-\hat\delta_{PK}^{\theta}(\bm x_i)\right)\frac{\partial\left(\delta_i^{\theta}-\hat\delta_{PK}^{\theta}(\bm x_i)\right)}{\partial\theta_j}+o_p(n^{-1/2})\\
&=\frac{\partial {\bm b^T_{\theta}}}{\partial \theta_j}\mathbf E^{-1}_{\theta}\frac{2}{n}\sum_{i=1}^n\delta_i^{\theta}\bm g_{\theta}(\bm x_i) +\bm b^T_{\theta}\frac{\partial {\mathbf E^{-1}_{\theta}}}{\partial \theta_j} \frac{2}{n}\sum_{i=1}^n \delta_i^{\theta}\bm g_{\theta}(\bm x_i)+\bm b^T_{\theta}\mathbf E^{-1}_{\theta}\frac{2}{n}\sum_{i=1}^n \delta_i^{\theta}\frac{\partial \bm g_{\theta}(\bm x_i)}{\partial \theta_j}\\
&-\bm b^T_{\theta}\mathbf E^{-1}_{\theta}\frac{2}{n}\sum_{i=1}^n\hat\delta_{PK}^{\theta}(\bm x_i)\frac{\partial \bm g_{\theta}(\bm x_i)}{\partial \theta_j}+o_p(n^{-1/2}).
   \end{aligned}
   \label{eqdpkc}
  \end{eqnarray}

Then we work on $\frac{1}{n}\sum_{i=1}^n \delta_i^{\theta}\bm g_{\theta}(\bm x_i)$, $\frac{1}{n}\sum_{i=1}^n \delta_i^{\theta}\frac{\partial \bm g_{\theta}(\bm x_i)}{\partial \theta_j}$ and $\frac{2}{n}\sum_{i=1}^n\hat\delta_{PK}^{\theta}(\bm x_i)\frac{\partial \bm g_{\theta}(\bm x_i)}{\partial \theta_j} $ separately. 

{\begin{itemize}
\item By the definition of $\delta_i^{\theta}$, it is easily obtained that 
\begin{eqnarray} 
  \begin{aligned}
\frac{1}{n}\sum_{i=1}^n \delta_i^{\theta}\bm g_{\theta}(\bm x_i)&=<\delta^{\theta},\bm g_{\theta}>_n+<\epsilon,\bm g_{\theta}>_n.
     \end{aligned}
  \end{eqnarray}
Because $\mathcal{N}_{K_{\theta}}(\Omega)$ can be continuously embedded into the Sobolev space $H^m(\Omega)$, it follows from Theorem 5.11 of \cite{geer2000empirical} that
$\frac{1}{\sqrt{n}}\sum_{i=1}^n\epsilon_i\bm g_{\theta}(\bm x_i)=O_p(1).$
 By combining with (\ref{semi2}), we have
  \begin{eqnarray} 
  \frac{1}{n}\sum_{i=1}^n \delta_i^{\theta}\bm g_{\theta}(\bm x_i)=<\delta^{\theta},\bm g_{\theta}>_{L_2(\Omega)}+O_p(n^{-1/2}).
  \label{intb}
    \end{eqnarray}

\item Similarly,  
      \begin{eqnarray} 
      \begin{aligned}
\frac{1}{n}\sum_{i=1}^n \delta_i^{\theta}\frac{\partial \bm g_{\theta}(\bm x_i)}{\partial \theta_j} &=<\delta^{\theta},\frac{\partial \bm g_{\theta}(\bm x_i)}{\partial \theta_j} >_n+<\epsilon,\frac{\partial \bm g_{\theta}(\bm x_i)}{\partial \theta_j} >_n,\\
&=\frac{\partial <\delta^{\theta},\bm g_{\theta}>_{L_2(\Omega)}}{\partial \theta_j}+< \frac{\partial y^s(\bm x,\bm\theta)}{\partial \theta_j},\bm g_{\theta}>_{{L_2(\Omega)}}+O_p(n^{-\frac{1}{2}}).
\end{aligned}
 \label{intb2}
    \end{eqnarray}

\item Because  $\int \hat\delta_{PK}^{\theta}\bm g_{\theta}d\bm x=0$, taking the partial derivative of $\int \hat\delta_{PK}^{\theta}\bm g_{\theta}d\bm x$ on $\theta_j$, we have
\begin{eqnarray*} 
  \begin{aligned}
\int  \frac{\partial  \hat \delta_{PK}^{\theta}}{\partial \theta_j}  \bm g_{\theta}d\bm x+\int  \hat \delta_{PK}^{\theta}\frac{\partial\bm g_{\theta} }{\partial \theta_j}d\bm x=0.
     \end{aligned}
  \end{eqnarray*}
Moverover, because $\int \bm g_{\theta}\bm g^T_{\theta}d\bm x=\mathbf E_{\theta}$, and $\frac{\partial {\mathbf E^{-1}_{\theta}}}{\partial \theta_j}=-\mathbf E^{-1}_{\theta}\frac{\partial {\mathbf E_{\theta}}}{\partial \theta_j}\mathbf E^{-1}_{\theta}$, we have
 \begin{eqnarray} 
  \begin{aligned}
\int  \hat \delta_{PK}^{\theta}\frac{\partial\bm g_{\theta} }{\partial \theta_j}d\bm x&=-\int \bm g_{\theta}  \frac{\partial \hat \delta_{PK}^{\theta}}{\partial \theta_j}d\bm x,\\
&= \int  \bm g_{\theta}\frac{\partial y^s(\bm x_i,\bm\theta)}{\partial \theta_j}d \bm x+\frac{\partial {\bm b_{\theta}}}{\partial \theta_j}+\frac{1}{2}\mathbf E_{\theta}\frac{\partial {\mathbf E^{-1}_{\theta}}}{\partial \theta_j} \bm b_{\theta}.
     \end{aligned}
     \label{inteq}
  \end{eqnarray}
  \end{itemize}}

{Recall that $\bm a^{T}_{\theta}=<\delta^{\theta}(\cdot),\frac{\partial y^s(\cdot,\bm\theta)}{\partial\bm\theta}>_{L_2(\Omega)}$.} By combining (\ref{intb}), (\ref{intb2}) and (\ref{inteq}), we have that (\ref{eqdpkc}) can be represented as 
     \begin{eqnarray}
     \label{dlpk}
     \begin{aligned}
2 \frac{\partial {\bm b^T_{\theta}}}{\partial \theta_j}\mathbf E^{-1}_{\theta}(\bm a_{\theta}-\bm b_{\theta})+
2 \frac{\partial {\bm a^T_{\theta}}}{\partial \theta_j}\mathbf E^{-1}_{\theta}\bm b_{\theta}+\bm a^T_{\theta}\frac{\partial {\mathbf E^{-1}_{\theta}}}{\partial \theta_j} (2\bm a_{\theta}-\bm b_{\theta})+O_p(n^{-\frac{1}{2}}).
     \end{aligned}
     \end{eqnarray}
     
    Let $\bm a_{\theta^s}=\bm 0$,  to check whether $\bm\theta^s$ is a local maximum or a local minimum of the PK loss function, the Hessian matrix of $L_{PK}(\bm\theta)$ at $\bm\theta^s$, denotes as $\mathbf H_{PK}(\bm\theta^s)$, can be evaluated from (\ref{dlpk}). Following from the Proposition  \ref{condelta}, together with the Cauchy-Schwarz inequality, we have that $\bm b_{\theta}=\bm a_{\theta}+O_p( n^{-\frac{m}{2m+d}})$, and thus
 \begin{eqnarray} 
    \begin{aligned}
\mathbf H_{PK}(\bm\theta^s)=2\frac{\partial {\boldsymbol a^T_{\boldsymbol\theta^s}}}{\partial{\bm \theta}}\mathbf E^{-1}_{\boldsymbol\theta^s}\frac{\partial {\boldsymbol a^T_{\boldsymbol\theta^s}}}{\partial{\bm \theta}}+O_p(n^{-\frac{m}{2m+d}}).
    \end{aligned}
    \end{eqnarray}

     That is,  the desired results are obtained.

 \end{proof}

\section{Technical proofs in Section \ref{sec4}}
In this section, we prove the theorems in Section \ref{sec4}.
\subsection{Proof of Theorem \ref{th:stationary points ppk} }
\begin{proof}
From (\ref{ppkloss}),  the stationary points of PPK loss function satisfy that
   \begin{eqnarray} 
    \begin{aligned}
0=& \frac{\partial}{\partial \theta_j} \left\{L_{PK}(\bm\theta)+\eta\|\hat\delta_{PK}^{\theta}\|^2_{L_2(\Omega)}\right\}.
    \end{aligned}
    \label{dppkc}
    \end{eqnarray}
   {The first derivative of the PK loss function on $\theta_j$ can be found in  (\ref{dlpk}). Next, we focus on  the first derivative of  $\|\hat\delta_{PK}^{\theta}\|^2_{L_2(\Omega)}$ on $\theta_j, j=1,\ldots,q$.}
 
  The partial derivative of $\|\hat\delta_{PK}^{\theta}\|^2_{L_2(\Omega)}$ on $\theta_j$ is
  \begin{eqnarray} 
  \label{b2}
\frac{ \partial\|\hat\delta_{PK}^{\theta}\|^2_{L_2(\Omega)}}{\partial \theta_j}=\frac{ \partial \int_{\Omega} (\hat\delta_{PK}^{\theta} )^2d\bm x}{\partial \theta_j}=2\int_{\Omega} \hat\delta_{PK}^{\theta} \frac{ \partial  \hat\delta_{PK}^{\theta} }{\partial \theta_j}d\bm x.
 \end{eqnarray}
 where $\frac{ \partial  \hat\delta_{PK}^{\theta} }{\partial \theta_j}$ is shown in (\ref{ddeltahat}). Because $\int \hat\delta_{PK}^{\theta}\bm g_{\theta}d\bm x=\bm 0$, we have
 \begin{eqnarray} 
  \label{b3}
 \begin{aligned}
&2\int_{\Omega} \hat\delta_{PK}^{\theta} \frac{ \partial  \hat\delta_{PK}^{\theta} }{\partial \theta_j}d\bm x\\=&2\int_{\Omega} \hat\delta_{PK}^{\theta}\times\left( -\frac{\partial y^s(\bm x,\bm\theta)}{\partial \theta_j}-\frac{\partial {\bm b^T_{\theta}}}{\partial \theta_j}\mathbf E^{-1}_{\theta}\bm g_{\theta}-\bm b^T_{\theta}\frac{\partial {\mathbf E^{-1}_{\theta}}}{\partial \theta_j} \bm g_{\theta}-\bm b^T_{\theta}\mathbf E^{-1}_{\theta}\frac{\partial \bm g_{\theta}}{\partial \theta_j}\right)d\bm x,\\
=&-2\bm b^T_{\theta}\mathbf E^{-1}_{\theta}\int_{\Omega}\hat\delta_{PK}^{\theta}\frac{\partial \bm g_{\theta}}{\partial \theta_j}d\bm x.
\end{aligned}
 \end{eqnarray}

The integration by parts formula suggests that $\frac{\partial \bm b_{\theta}}{\partial\theta_j}=\int_{\Omega}\hat\delta^{\theta}\frac{\partial \bm g_{\theta}}{\partial \theta_j}d\bm x-\int_{\Omega}\frac{\partial y^s(\bm x,\bm\theta)}{\partial \theta_j}  \bm g_{\theta}d\bm x$, and $ \frac{\partial {\mathbf E_{\theta}}}{\partial \theta_j}= 2\int \frac{\partial \bm g_{\theta}}{\partial \theta_j}\bm g^T_{\theta}d\bm x$. The derivative of inverse matrix shows that $\frac{\partial {\mathbf E^{-1}_{\theta}}}{\partial \theta_j}=-\mathbf E^{-1}_{\theta}\frac{\partial {\mathbf E_{\theta}}}{\partial \theta_j}\mathbf E^{-1}_{\theta}$. Thus we have
 \begin{eqnarray} 
 \label{b4}
  \begin{aligned}
2\int  \hat \delta_{PK}^{\theta}\frac{\partial\bm g_{\theta} }{\partial \theta_j}d\bm x&=-2\bm b^T_{\theta}\mathbf E^{-1}_{\theta}\frac{\partial \bm b_{\theta}}{\partial\theta_j}-2\bm b^T_{\theta}\mathbf E^{-1}_{\theta}\int_{\Omega}\frac{\partial y^s(\bm x,\bm\theta)}{\partial \theta_j}  \bm g_{\theta}d\bm x-\bm b^T_{\theta}\frac{\partial {\mathbf E^{-1}_{\theta}}}{\partial \theta_j}\bm b_{\theta}.\\
     \end{aligned}
  \end{eqnarray}

  Combining  (\ref{dlpk}) with (\ref{b2})-(\ref{b4}), we have that  (\ref{dppkc}) can be represented as 
     \begin{eqnarray}
     \begin{aligned}
&(1-\eta) \left[2 \frac{\partial {\bm a^T_{\theta}}}{\partial \theta_j}\mathbf E^{-1}_{\theta}\bm a_{\theta}+\bm a^T_{\theta}\frac{\partial {\mathbf E^{-1}_{\theta}}}{\partial \theta_j} \bm a_{\theta}\right]-2\eta\bm a^T_{\theta}\mathbf E^{-1}_{\theta}\int_{\Omega}\frac{\partial y^s(\bm x,\bm\theta)}{\partial \theta_j}  \bm g_{\theta}d\bm x+O_p(n^{-\frac{m}{2m+d}}).
     \end{aligned}
     \end{eqnarray}
     
Now we compare $\eta$ with $1$, and consider two different cases.

Case I. If  $\eta=1$, because $\mathbf E_{\theta}= \int  \bm g_{\theta}\bm g^T_{\theta}d\bm x$.  then we have
      \begin{eqnarray}
     \begin{aligned}
 \frac{\partial}{\partial \bm \theta} \left\{L_{PK}(\bm\theta)+\eta\|\hat\delta_{PK}^{\theta}\|^2_{L_2(\Omega)}\right\}
&=-2\bm a^T_{\theta}+O_p(n^{-\frac{m}{2m+d}}).
     \end{aligned}
     \end{eqnarray}
   Recall that $\bm\theta^s$ is a stationary point of the $L_2$ loss function, which satisfy that,
 \begin{eqnarray}
    \begin{aligned}
   \int\left(\zeta(\bm x)-y^s(\bm x,\bm\theta^s)\right)\frac{\partial y^s(\bm x,\bm\theta^s)}{\partial \bm\theta}d\bm x=\bm a^{T}_{\bm\theta^s}=\bm {0}.
        \end{aligned}
   \end{eqnarray} 
We have $\bm\theta^s$ is a stationary point of the PPK loss function. The Hessian matrix  at $\bm\theta^s$ can easily obtained by evaluating second derivative of the PPK loss on $\bm\theta$,
\begin{eqnarray}
     \begin{aligned}
 \frac{\partial^2}{\partial \bm \theta \partial \bm \theta^{T}} \left\{L_{PK}(\bm\theta)+\eta\|\hat\delta_{PK}^{\theta}\|^2_{L_2(\Omega)}\right\}|_{\bm\theta=\bm\theta^s}
&=-2\left(\frac{\partial \bm a_{\theta}}{ \partial \bm \theta}\right)^T|_{\bm\theta=\bm\theta^s}+O_p(n^{-\frac{m}{2m+d}}).
     \end{aligned}
     \end{eqnarray}
     Moreover, Hessian matrix  at $\bm\theta^s$ of the $L_2$ loss function is 
     \begin{eqnarray}
         \begin{aligned}
 \frac{\partial^2}{\partial \bm \theta \partial \bm \theta^{T}} \left\{\|\zeta(\cdot)-y^s(\cdot,\bm\theta)\|^2_{L_2(\Omega)}\right\}|_{\bm\theta=\bm\theta^s}
&=-2\left(\frac{\partial \bm a_{\theta}}{\partial\bm\theta}\right)^{T}.
     \end{aligned}
     \end{eqnarray}
     
It is easily to be proven that, $\frac{\partial \bm a_{\theta}}{\partial\bm \theta}=\lim _{n\rightarrow \infty}\frac{\partial \bm b_{\theta}}{\partial\bm \theta}$. By the order-preserving {property} of limits of real sequences, we have that if  $ \bm\theta^s$ is a local minimum (maximum) of  the $L_2$ loss function, then  $ \bm\theta^s$ is a local minimum (maximum) of  the PPK loss function.  

Case II. If  $\eta\neq 1$, then 
\begin{eqnarray}
     \begin{aligned}
 &\frac{\partial^2}{\partial \bm \theta\partial \bm \theta^{T}} \left\{L_{PK}(\bm\theta)+\eta\|\hat\delta_{PK}^{\theta}\|^2_{L_2(\Omega)}\right\}|_{\bm\theta=\bm\theta^s}\\
 =&(1-\eta) \left(\frac{\partial \bm a_{\theta}}{\partial \bm\theta} \right)^T\mathbf E^{-1}_{\theta}\frac{\partial {\bm a_{\theta}}}{\partial \bm\theta}|_{\bm\theta=\bm\theta^s}-2\eta\left(\frac{\partial \bm a_{\theta}}{\partial \bm\theta} \right)^T|_{\bm\theta=\bm\theta^s}+O_p(n^{-\frac{m}{2m+d}}),\\
 =&(1-\eta) \left(\frac{\partial \bm a_{\theta}}{\partial \bm\theta} \right)^T|_{\bm\theta=\bm\theta^s}\mathbf E^{-1}_{\theta^s}\left\{\frac{\partial {\bm a_{\theta}}}{\partial \bm\theta}-2\eta/(1-\eta)\mathbf E_{\theta}\right\}|_{\bm\theta=\bm\theta^s}+O_p(n^{-\frac{m}{2m+d}}),\\
     \end{aligned}
     \end{eqnarray}
Now, we want to find the interval of $\eta$ such that if $\frac{\partial {\bm a_{\theta}}}{\partial \bm\theta}|_{\bm\theta=\bm\theta^s}$ is positive (negative) semi-definite, then $\frac{\partial^2}{\partial \bm \theta\partial \bm \theta^{T}} \left\{L_{PPK}(\bm\theta)\right\}|_{\bm\theta=\bm\theta^s}$ is negative (positive) semi-definite. That is, the product of $(1-\eta)$ and $\left\{\frac{\partial {\bm a_{\theta}}}{\partial \bm\theta}-2\eta/(1-\eta)\mathbf E_{\theta}\right\}|_{\bm\theta=\bm\theta^s}$ is negative semi-definite.

Suppose there exist constants $U\geq L$, such that
$$U\mathbf E_{\theta^s}>\frac{\partial {\bm a_{\theta}}}{\partial \bm\theta}|_{\bm\theta=\bm\theta^s}>L\mathbf E_{\theta^s}.$$
Then we have 
\begin{eqnarray}
\label{inequalitye}
     \begin{aligned}
(U-\frac{2\eta}{1-\eta})\mathbf E_{\theta}>\left\{\frac{\partial {\bm b_{\theta}}}{\partial \bm\theta}-\frac{2\eta}{1-\eta}\mathbf E_{\theta}\right\}|_{\bm\theta=\bm\theta^s}>(L-\frac{2\eta}{1-\eta})\mathbf E_{\theta}.
     \end{aligned}
     \end{eqnarray}

To evaluating the production of $1-\eta$ and (\ref{inequalitye}), We consider the sign of $1-\eta$.

\begin{enumerate}[(A)]
 \item If   $(1-\eta)>0$, then $(U-\frac{2\eta}{1-\eta})\leq 0$ is needed to guarantee that  $\frac{\partial^2}{\partial \bm \theta\partial \bm \theta^{T}} \left\{L_{PPK}(\bm\theta)\right\}|_{\bm\theta=\bm\theta^s}$ is negative semi-definite. That is ${U}\leq(U+2)\eta$ and $\eta<1$.
 \item If   $(1-\eta)<0$, then $(L-\frac{2\eta}{1-\eta})\geq 0$ is needed. That is ${L}\leq (L+2)\eta$ and $\eta>1$.
\end{enumerate}

By combining Case I and Case II, we have $\eta$ belongs to the set  $\Gamma_{\eta}=\{\eta=1\} {\cup}\{{U}\leq (U+2)\eta$ \text{and} $\eta<1\} {\cup} \{{L}\leq (L+2)\eta \text{ and }\eta>1\}$. By some easy calculations, $\Gamma_{\eta}$ can be represented as
 \begin{equation}
 \label{gamma}
\Gamma_{\eta}=\left\{
\begin{array}{rcl}
0\leq \eta< \frac{L}{L+2} & & { L< U\leq-2}  \text{ or }  L= U<-2 \\
\max(0,\frac{U}{U+2} )\leq \eta\leq \frac{L}{L+2}  & & {L < -2< U}\\
\eta>\max(0,\frac{U}{U+2} )& & {-2\leq L< U} \text{ or } L= U>-2 \\
\eta\geq 0 & &  { L= U=-2} 
\end{array} \right.
\end{equation}
\end{proof}

\subsection{Proof of Theorem \ref{consistppk}}
\begin{proof}
{We first  prove that $\hat{\bm\theta}_{PPK}$ converges  to $\bm\theta^*$ in probability. The desired results can be proved by showing that
  \begin{eqnarray}
  L_{PPK}(\bm\theta^*)\leq \inf_{\|\bm\theta-\bm\theta^*\|=cn^{-\frac{m}{2m+d}}}L_{PPK}(\bm\theta),
  \label{basic ineq}
   \end{eqnarray}
  for sufficiently large $n$ and some constant $c>0$ to be specified later, where $\|\cdot\|$ denotes the usual Euclidean distance. 
  Then we prove that $\hat{\bm\theta}_{PPK}$ converges in distribution to a normal distribution by following the standard framework for establishing asymptotic theory for M-estimation.
  }
  
  We use the converse method to prove that (\ref{basic ineq}) holds.
  Suppose (\ref{basic ineq}) is false. Then there exists $\tilde{\bm\theta}$ with $\|\bm\theta^*-\tilde{\bm\theta}\|=cn^{-\frac{m}{2m+d}}$ so that
    \begin{eqnarray}
    \label{ineq1}
    \begin{aligned}
 L_{PK}(\bm\theta^*)+\eta\|\hat\delta_{PK}^{\theta^*}\|^2_{L_2(\Omega)}>  L_{PK}(\tilde{\bm\theta})+\eta\|\hat\delta_{PK}^{\tilde{\theta}}\|^2_{L_2(\Omega)}.
 \end{aligned}
  \end{eqnarray}
{Because the sequence $\{\hat{\bm\theta}_{PK}\}$ converges to $\bm\theta^*$ in probability as $n$ goes to infinity, by the proof of Theorem 4.2 in \cite{tuo2019adjustments}, we have that}
   \begin{eqnarray}
   \label{lpkeq1}
  L_{PK}(\bm\theta^*)\leq \inf_{\|\bm\theta-\bm\theta^*\|=c_1n^{-\frac{m}{2m+d}}}L_{PK}(\bm\theta),
   \end{eqnarray}
    for sufficiently large $n$ and some constant $c_1>0$.  Let $c=c_1$,   (\ref{lpkeq1}) implies that
  \begin{eqnarray}
  L_{PK}(\bm\theta^*)\leq L_{PK}(\tilde{\bm\theta}).
  \label{basic ineq pk}
   \end{eqnarray}
    Combining (\ref{ineq1}) and (\ref{basic ineq pk}),  we arrive at
  \begin{eqnarray}
 \|\hat\delta_{PK}^{\theta^*}\|^2_{L_2(\Omega)}>\|\hat\delta_{PK}^{\tilde{\theta}}\|^2_{L_2(\Omega)}.
   \end{eqnarray}
  
   On the other hand, by the definition of $\bm\theta^*$, we have that
   \begin{eqnarray}
 \|\delta^{\theta^*}\|^2_{L_2(\Omega)}\leq \|\delta^{\tilde{\theta}}\|^2_{L_2(\Omega)}.
   \end{eqnarray}
By the uniform convergence $\hat\delta_{PK}^{\theta}$ in Proposition \ref{condelta}, we have 
 \begin{eqnarray}
\lim_{n\rightarrow \infty}  \|\hat\delta_{PK}^{\theta^*}\|^2_{L_2(\Omega)}=\|\delta^{\theta^*}\|^2_{L_2(\Omega)}.
   \end{eqnarray}
By the order-preserving properties of limits of real sequences, there exist $N\in \bm{R}$ such that, for any $n>N$,
 \begin{eqnarray}
 \|\hat\delta_{PK}^{\theta^*}\|^2_{L_2(\Omega)}\leq\|\hat\delta_{PK}^{\tilde{\theta}}\|^2_{L_2(\Omega)}.
   \end{eqnarray}
   This leads to a contradiction.
   
{Next we prove that $\hat{\bm\theta}_{PPK}$ converges in distribution to a normal distribution. Because $\hat\delta^{\hat{\bm\theta}_{PPK}}_{PK}=\mathcal{P}^{\perp}_{\mathcal{G}_{\hat\theta_{PPK}}}{\hat\delta^{\hat{\theta}_{PPK}}}$, by the definition of $\hat{\bm\theta}_{PPK}$, we have 
 \begin{eqnarray}
 \label{asynorm1}
 \frac{\partial \tilde l(\bm\theta)}{\partial \bm\theta}|_{\bm\theta=\hat{\bm\theta}_{PPK}}=0,
    \end{eqnarray}
where
$$\tilde l(\bm\theta)=\frac{1}{n}\sum_{i=1}^n\left(y_i-y^s(\bm x_i,\bm\theta)-\mathcal{P}^{\perp}_{\mathcal{G}_{\theta}}\hat\delta^{\hat{\theta}_{PPK}}(\boldsymbol x_i)\right)^2+\lambda\|\mathcal{P}^{\perp}_{\mathcal{G}_{\theta}}\hat\delta^{\hat{\theta}_{PPK}}\|^2_{\mathcal{N}_{K_{\theta}}(\Omega)}+\eta\|\mathcal{P}^{\perp}_{\mathcal{G}_{\theta}}\hat\delta^{\hat{\theta}_{PPK}}\|^2_{L_2(\Omega)}.$$

Involved with $\mathcal {P}_{\mathcal{G}_{\theta}}\hat\delta^{\hat{\theta}_{PPK}}=\bm {\tilde b}^{T}_{\theta}\mathbf E^{-1}_{\theta}\bm g_{\theta}$, where $\bm {\tilde b}_{\theta}=<\hat\delta^{\hat{\theta}_{PPK}}, \bm g_{\theta}>_{L_2(\Omega)}$,  we have that
$$\frac{\partial \mathcal {P}^{\perp}_{\mathcal{G}_{\theta}}\hat\delta^{\hat{\theta}_{PPK}}}{\partial \bm\theta}=-\frac{\partial \bm {\tilde b}^{T}_{\theta}}{\partial\bm\theta}\mathbf E^{-1}_{\theta}\bm g_{\theta}-\bm {\tilde b}^{T}_{\theta}\frac{\partial \mathbf E^{-1}_{\theta}\bm g_{\theta}}{\partial\bm\theta}.$$
Because $<\mathcal {P}^{\perp}_{\mathcal{G}_{\theta}}\hat\delta^{\hat{\theta}_{PPK}},\bm g_{\theta}>|_{\bm\theta=\hat{\bm\theta}_{PPK}}=\bm 0$, we can easily verify that
$$\frac{\partial \mathcal {P}^{\perp}_{\mathcal{G}_{\theta}}\hat\delta^{\hat{\theta}_{PPK}}}{\partial \bm\theta}|_{\bm\theta=\hat{\bm\theta}_{PPK}}=-\frac{\partial \bm {\tilde b}^{T}_{\theta}}{\partial\bm\theta}\mathbf E^{-1}_{\theta}\bm g_{\theta}|_{\bm\theta=\hat{\bm\theta}_{PPK}},$$
and
$$\frac{\partial \|\mathcal{P}^{\perp}_{\mathcal{G}_{\theta}}\hat\delta^{\hat{\theta}_{PPK}}\|^2_{L_2(\Omega)}}{\partial \bm\theta}|_{\bm\theta=\hat{\bm\theta}_{PPK}}=\bm 0.$$
Moreover,  Eq. (69) in \cite{tuo2019adjustments} shows that
$$\lambda \frac{\partial \|\mathcal{P}^{\perp}_{\mathcal{G}_{\theta}}\hat\delta^{\hat{\theta}_{PPK}}\|^2_{\mathcal{N}_{K_{\theta}}(\Omega)}}{\partial \bm\theta}|_{\bm\theta=\hat{\bm\theta}_{PPK}}=o_p(n^{-1/2}).$$
 By some elementary calculations, (\ref{asynorm1}) becomes 
 \begin{eqnarray}
 \label{asynorm2}
  \frac{1}{n}\sum_{i=1}^n\left[(\zeta(\bm x_i)-y^s(\bm x_i,\hat{\bm\theta}_{PPK}))\bm g_{\hat\theta_{PPK}}\right]+ \frac{1}{n} \sum_{i=1}^n\epsilon_i \bm g_{\hat\theta_{PPK}}=o_p(n^{-1/2}).
    \end{eqnarray}
By applying Taylor's theorem, the first part of  (\ref{asynorm2}) can be represented by
 \begin{eqnarray}
 \label{asynorm2-1}
 \begin{aligned}
 \frac{1}{n}\sum_{i=1}^n\left[(\zeta(\bm x_i)-y^s(\bm x_i,{\bm\theta^*}))\bm g_{\theta^*}\right] -\frac{1}{2}\mathbf{V}(\hat {\bm\theta}_{PPK}-\bm\theta^*)+o_p(n^{-1/2}).
    \end{aligned}
     \end{eqnarray}
Because $\delta^*= {\delta}^{\theta^*}\in{\mathcal{G}^{\perp}_{\theta^*}}$, together with the asymptotic equivalence relation between the $L_2$ and the empirical norm (\ref{semi2}),  there is $\frac{1}{n}\sum_{i=1}^n\left[(\zeta(\bm x_i)-y^s(\bm x_i,{\bm\theta^*}))\bm g_{\theta^*}\right]= o_p(n^{-1/2})$. The desired result can be then obtained by combining  (\ref{asynorm2}) and ( \ref{asynorm2-1}).
}
 \end{proof}  
 \subsection{Proof of Theorem \ref{predpower}}
 \begin{proof}
  Let $\hat\zeta_{n}(\cdot)=\hat\delta^{\hat{\bm\theta}_{PPK}}_{PK}(\cdot)+y^s(\cdot,\hat{\bm\theta}_{PPK})$, which is an estimator of $\zeta(\cdot)$. The triangle inequality implies that
{  \begin{eqnarray} 
\begin{aligned}
\|\hat\zeta_n-\zeta\|_{L_2(\Omega)}
\leq& \|\hat\delta^{{\bm\theta^*}}_{PK}-\delta^*\|_{L_2(\Omega)}+\\
&\|\hat\delta^{\hat{\bm\theta}_{PPK}}_{PK}-\hat\delta^{{\bm\theta^*}}_{PK}\|_{L_2(\Omega)}+\|y^s(\cdot,\hat{\bm\theta}_{PPK})-y^s(\cdot,{\bm\theta^*})\|_{L_2(\Omega)},\\
=&\rm{I+II+III}.
\end{aligned}
\end{eqnarray} }
Next we bound (I),(II) and (III) respectively. 

{For (I), because $\delta^*= {\delta}^{\theta^*}\in{\mathcal{G}^{\perp}_{\theta^*}}$, it follows from Corollary \ref{col3.6} that 
$${\rm{I}}\leq \sup_{\boldsymbol \theta\in\Theta}\left\|\mathcal{P}^{\perp}_{\mathcal{G}_{\theta}} {\delta}^{\theta}-\hat{\delta}^{\theta}_{PK}\right\|_{L_2(\Omega)} =O_p(n^{-\frac{m}{2m+d}}).$$ 

For  (II), we can apply Taylor's theorem to conclude that
  \begin{eqnarray} 
\begin{aligned}
 \|\hat\delta^{\hat{\bm\theta}_{PPK}}_{PK}-\hat\delta^{{\bm\theta^*}}_{PK}\|_{L_2(\Omega)}\leq &\left\|\frac{\partial\hat\delta^{\theta^*}_{PK}}{\partial\bm\theta}\right\|_{L_2(\Omega)}\|\hat{\bm\theta}_{PPK}-\bm\theta^*\|+O_p(\|\hat{\bm\theta}_{PPK}-\bm\theta^*\|^2).
  \end{aligned}
\end{eqnarray}
Here, $\|\cdot\|$ is the euclidean distance. 
The first derivative of $\hat\delta_{PK}^{\theta}$ on $\theta_j, j=1,\ldots,q$ (\ref{ddeltahat}) suggests that 
 \begin{eqnarray} 
  \begin{aligned}
  \frac{\partial \hat \delta_{PK}^{\theta}}{\partial \theta_j}|_{\bm\theta^*}
 &= -\frac{\partial y^s(\bm x,\bm\theta^*)}{\partial \theta_j}-\frac{\partial {\bm b^T_{\theta^*}}}{\partial \theta_j}\mathbf E^{-1}_{\theta^*}\bm g_{\theta^*}(\bm x)-\\
 &-\bm b^T_{\theta^*}\frac{\partial {\mathbf E^{-1}_{\theta^*}}}{\partial \theta_j} \bm g_{\theta^*}(\bm x)-\bm b^T_{\theta^*}\mathbf E^{-1}_{\theta^*}\frac{\partial \bm g_{\theta^*}(\bm x)}{\partial \theta_j},\\
 &= -<\delta^*,\frac{\partial^2 y^s(\bm x,\bm\theta^*)}{\partial \theta_{i}\partial \theta_j}>_{L_2(\Omega)}\mathbf E^{-1}_{\theta^*}\bm g_{\theta^*}(\bm x)+O_p(n^{-\frac{m}{2m+d}}).
\end{aligned}
\label{ddeltahat2}
\end{eqnarray}

The last equality is following from Proposition  \ref{condelta}. By condition A3 and condition B2,  together with the triangle inequality, we have that $\|\frac{\partial \hat \delta_{PK}^{\theta^*}}{\partial \theta_j}\|_{L_2(\Omega)}$ can be bounded by a finite constant. Combining with the asymptotic normal of $\hat\theta_{PPK}$, we have that $ {\rm II}=O_p(n^{-1/2})$. }

Following a similar argument, we have  ${\rm{III}}=O_p(n^{-1/2})$. This leads to the desired result.

\end{proof}

\bibliographystyle{siamplain}
\bibliography{allrefs}

\end{document}

%% file: tmp_ppkc_header.tex
	\title{Penalized Projected Kernel Calibration for Computer Models\thanks{Submitted to the editors DATE.
			\funding{Dr. Wang's research was supported by the National Natural Science
Foundation of China (12101024), the Natural Science Foundation
of Beijing Municipality (1214019).}}}

	\author{Yan Wang \thanks{School of Statistics and Data Science, Faculty of Science, Beijing University of Technology, Beijing 100124, China (\email{yanwang@bjut.edu.cn}).
	}	
	}

	\headers{Penalized Projected Kernel Calibration}{ Yan Wang }

%% file: ppkc.bbl
\begin{thebibliography}{10}

\bibitem{bayarri2012framework}
{\sc M.~J. Bayarri, J.~O. Berger, R.~Paulo, J.~Sacks, J.~A. Cafeo,
  J.~Cavendish, C.-H. Lin, and J.~Tu}, {\em A framework for validation of
  computer models}, Technometrics, 49 (2007).

\bibitem{carothers2000real}
{\sc N.~L. Carothers}, {\em Real analysis}, Cambridge University Press, 2000.

\bibitem{gu2018robustgasp}
{\sc M.~Gu, J.~Palomo, and J.~O. Berger}, {\em Robustgasp: Robust {G}aussian
  stochastic process emulation in {R}}, arXiv preprint arXiv:1801.01874,
  (2018).

\bibitem{gu2018scaled}
{\sc M.~Gu and L.~Wang}, {\em Scaled {G}aussian stochastic process for computer
  model calibration and prediction}, SIAM/ASA Journal on Uncertainty
  Quantification, 6 (2018), pp.~1555--1583.

\bibitem{gu2018theoretical}
{\sc M.~Gu, F.~Xie, and L.~Wang}, {\em A theoretical framework of the scaled
  {G}aussian stochastic process in prediction and calibration}, arXiv preprint
  arXiv:1807.03829,  (2018).

\bibitem{gurbuz2017some}
{\sc F.~G{\"u}rb{\"u}z}, {\em Some estimates for generalized commutators of
  rough fractional maximal and integral operators on generalized weighted
  {M}orrey spaces}, Canadian Mathematical Bulletin, 60 (2017), pp.~131--145.

\bibitem{gurbuz2018generalized}
{\sc F.~G{\"u}rb{\"u}z}, {\em Generalized local {M}orrey spaces and multilinear
  commutators generated by marcinkiewicz integrals with rough kernel associated
  with schr{\"o}dinger operators and local campanato functions}, Journal of
  Applied Analysis \& Computation, 8 (2018), pp.~1369--1384.

\bibitem{gurbuz2020generalized}
{\sc F.~G{\"u}rb{\"u}z}, {\em Generalized weighted {M}orrey estimates for
  {M}arcinkiewicz integrals with rough kernel associated with schr{\"o}dinger
  operator and their commutators}, Chinese Annals of Mathematics, Series B, 41
  (2020), pp.~77--98.

\bibitem{gurbuz2020behaviors}
{\sc F.~G{\"u}rb{\"u}z}, {\em On the behaviors of rough multilinear fractional
  integral and multi-sublinear fractional maximal operators both on product
  ${L}_p$ and weighted ${L}_p$ spaces}, International Journal of Nonlinear
  Sciences and Numerical Simulation, 21 (2020), pp.~715--726.

\bibitem{gurbuz2021note}
{\sc F.~G{\"u}rb{\"u}z}, {\em A note concerning {M}arcinkiewicz integral with
  rough kernel}, Infinite Dimensional Analysis, Quantum Probability and Related
  Topics, 24 (2021), p.~2150005 (14 pages).

\bibitem{james2013bias}
{\sc G.~James, D.~Witten, T.~Hastie, and R.~Tibshirani}, {\em Bias-variance
  trade-off for {K}-fold cross-validation. an introd. to stat. learn.-with
  appl. r}, 2013.

\bibitem{kennedy2001bayesian}
{\sc M.~C. Kennedy and A.~O'Hagan}, {\em Bayesian calibration of computer
  models}, Journal of the Royal Statistical Society: Series B (Statistical
  Methodology), 63 (2001), pp.~425--464.

\bibitem{kuhn2015short}
{\sc M.~Kuhn}, {\em A short introduction to the caret package}, R Found Stat
  Comput, 1 (2015).

\bibitem{park1991tuning}
{\sc J.~S. Park}, {\em Tuning complex computer codes to data and optimal
  designs}, PhDT,  (1991).

\bibitem{plumlee2017bayesian}
{\sc M.~Plumlee}, {\em Bayesian calibration of inexact computer models},
  Journal of the American Statistical Association, 112 (2017), pp.~1274--1285.

\bibitem{powell2006newuoa}
{\sc M.~J. Powell}, {\em The newuoa software for unconstrained optimization
  without derivatives}, in Large-scale nonlinear optimization, Springer, 2006,
  pp.~255--297.

\bibitem{santner2013design}
{\sc T.~J. Santner, B.~J. Williams, and W.~I. Notz}, {\em The Design and
  Analysis of Computer Experiments}, Springer Science \& Business Media, 2013.

\bibitem{scholkopf2001generalized}
{\sc B.~Sch{\"o}lkopf, R.~Herbrich, and A.~J. Smola}, {\em A generalized
  representer theorem}, in International Conference on Computational Learning
  Theory, Springer, 2001, pp.~416--426.

\bibitem{soetaert2009rootsolve}
{\sc K.~Soetaert}, {\em rootsolve: Nonlinear root finding, equilibrium and
  steady-state analysis of ordinary differential equations}, R package version,
  1 (2009).

\bibitem{stein2012interpolation}
{\sc M.~L. Stein}, {\em Interpolation of Spatial Data: Some Theory for
  Kriging}, Springer Science \& Business Media, 1999.

\bibitem{stoer2013introduction}
{\sc J.~Stoer and R.~Bulirsch}, {\em Introduction to numerical analysis},
  vol.~12, Springer Science \& Business Media, 2013.

\bibitem{stone1982optimal}
{\sc C.~J. Stone}, {\em Optimal global rates of convergence for nonparametric
  regression}, The annals of statistics,  (1982), pp.~1040--1053.

\bibitem{talbi2009metaheuristics}
{\sc E.-G. Talbi}, {\em Metaheuristics: from design to implementation},
  vol.~74, John Wiley \& Sons, 2009.

\bibitem{tuo2019adjustments}
{\sc R.~Tuo}, {\em Adjustments to computer models via projected kernel
  calibration}, SIAM/ASA Journal on Uncertainty Quantification, 7 (2019),
  pp.~553--578.

\bibitem{tuo2020improved}
{\sc R.~Tuo, Y.~Wang, and C.~F. Jeff~Wu}, {\em On the improved rates of
  convergence for {M}at$\backslash$'ern-type kernel ridge regression with
  application to calibration of computer models}, SIAM/ASA Journal on
  Uncertainty Quantification, 8 (2020), pp.~1522--1547.

\bibitem{tuo2015efficient}
{\sc R.~Tuo and C.~F.~J. Wu}, {\em Efficient calibration for imperfect computer
  models}, The Annals of Statistics, 43 (2015), pp.~2331--2352.

\bibitem{tuo2016theoretical}
{\sc R.~Tuo and C.~F.~J. Wu}, {\em A theoretical framework for calibration in
  computer models: parametrization, estimation and convergence properties},
  SIAM/ASA Journal on Uncertainty Quantification, 4 (2016), pp.~767--795.

\bibitem{geer2000empirical}
{\sc S.~A. van~de Geer}, {\em Empirical Processes in M-estimation}, vol.~6,
  Cambridge university press, 2000.

\bibitem{wahba1990spline}
{\sc G.~Wahba}, {\em Spline Models for Observational Data}, vol.~59, Siam,
  1990.

\bibitem{wang2020effective}
{\sc Y.~Wang, X.~Yue, R.~Tuo, J.~H. Hunt, and J.~Shi}, {\em Effective model
  calibration via sensible variable identification and adjustment, with
  application to composite fuselage simulation}, Annals of Applied Statistics,
  14 (2020), pp.~1759--1776.

\bibitem{wendland2004scattered}
{\sc H.~Wendland}, {\em Scattered Data Approximation}, vol.~17, Cambridge
  university press, 2004.

\bibitem{wong2014frequentist}
{\sc R.~K. Wong, C.~B. Storlie, and T.~Lee}, {\em A frequentist approach to
  computer model calibration}, Journal of the Royal Statistical Society: Series
  B (Statistical Methodology), 79 (2017), pp.~635--648.

\bibitem{xie2020bayesian}
{\sc F.~Xie and Y.~Xu}, {\em Bayesian projected calibration of computer
  models}, Journal of the American Statistical Association,  (2020), pp.~1--18.

\bibitem{xiong2013sequential}
{\sc S.~Xiong, P.~Z. Qian, and C.~J. Wu}, {\em Sequential design and analysis
  of high-accuracy and low-accuracy computer codes}, Technometrics, 55 (2013),
  pp.~37--46.

\end{thebibliography}
